\documentclass[a4paper]{article}
\usepackage{authblk}
\usepackage{amsthm}
\usepackage[letterpaper, left=1in, right=1in, bottom=1in, top=1in]{geometry}

\usepackage{graphicx}
\usepackage{amsmath,amssymb}
\usepackage{booktabs}
\usepackage{multirow}

\newtheorem{theorem}{Theorem}
\newtheorem{corollary}{Corollary}
\newtheorem{lemma}{Lemma}
\newtheorem{observation}{Observation}
\newtheorem{problem}{Problem}

\newcommand{\loglog}{\log\log}

\newcommand{\STree}{\mathsf{STree}}
\newcommand{\str}{\mathit{str}}

\newcommand{\Rarm}{\mathsf{Rarm}}
\newcommand{\Larm}{\mathsf{Larm}}
\newcommand{\rarm}{\mathsf{rarm}}
\newcommand{\larm}{\mathsf{larm}}

\newcommand{\MUPS}{\mathsf{MUPS}}
\newcommand{\beg}{\mathit{beg}}
\newcommand{\inbeg}{\mathit{inbeg}}
\newcommand{\xbeg}{\mathit{xbeg}}
\newcommand{\sub}{\mathit{sub}}
\newcommand{\add}{\mathit{add}}
\newcommand{\rem}{\mathit{rem}}
\newcommand{\MA}{\mathsf{MA}}
\newcommand{\EARM}{\mathit{EA}}
\newcommand{\lcp}{\mathit{lcp}}
\newcommand{\lce}{\mathit{lce}}
\newcommand{\nca}{\mathsf{nca}}
\newcommand{\EERTREE}{\mathsf{EERTREE}}
\newcommand{\Todd}{\mathcal{T}_{\mathsf{odd}}}
\newcommand{\Teven}{\mathcal{T}_{\mathsf{even}}}
\newcommand{\pal}{\mathit{pal}}

\newcommand{\odd}{\mathsf{odd}}
\newcommand{\even}{\mathsf{even}}
\newcommand{\centerpal}{\mathit{center}}

\newcommand{\Ext}{\mathit{Ext}}

\graphicspath{{img/}}

\usepackage{url}
\usepackage{hyperref}
\usepackage{cleveref}
\usepackage[final]{fixme}
\fxsetup{theme=color,mode=multiuser}
\FXRegisterAuthor{tm}{atm}{\color{red}TM}
\FXRegisterAuthor{mf}{amf}{MF}
\fxsetface{env}{}

\begin{document}
\title{
  \renewcommand*{\thefootnote}{\fnsymbol{footnote}}
  Minimal unique palindromic substrings\\after single-character substitution\footnote{This is a full version of a paper presented at SPIRE 2021~\cite{FunakoshiM21MUPSEdit}.}}
\author[1,2]{Mitsuru~Funakoshi}
\author[1,2]{Takuya~Mieno}
\affil[1]{Department of Informatics, Kyushu University.
\texttt{\{{mitsuru.funakoshi,takuya.mieno\}@inf.kyushu-u.ac.jp}}
}
\affil[2]{Japan Society for the Promotion of Science.}
\date{}
\maketitle
\begin{abstract}
A \emph{palindrome} is a string that reads the same forward and backward.
  A palindromic substring $w$ of a string $T$ is called a \emph{minimal unique palindromic substring}~(\emph{MUPS}) of $T$
  if $w$ occurs only once in $T$ and any proper palindromic substring of $w$ occurs at least twice in $T$.
MUPSs are utilized for answering the \emph{shortest unique palindromic substring problem}, which is motivated by molecular biology~[Inoue et al., 2018].
  Given a string $T$ of length $n$, all MUPSs of $T$ can be computed in $O(n)$ time.
In this paper, we study the problem of updating the set of MUPSs when a character in the input string $T$ is substituted by another character.
  We first analyze the number $d$ of changes of MUPSs when a character is substituted, and show that $d$ is in $O(\log n)$.
Further, we present an algorithm that uses $O(n)$ time and space for preprocessing,
and updates the set of MUPSs in $O(\log\sigma + (\loglog n)^2 + d)$ time where $\sigma$ is the alphabet size.
We also propose a variant of the algorithm, which runs in optimal $O(1+d)$ time when the alphabet size is constant.
\end{abstract}
 \section{Introduction}
Palindromes are strings that read the same forward and backward.
Finding palindromic structures has important applications to analyze biological data such as DNA, RNA, and proteins,
and thus algorithms and combinatorial properties on palindromic structures have been heavily studied~(e.g., see~\cite{Manacher75,Apostolico1995parallel,DroubayJP01,matsubara_tcs2009,Rubinchik2018eertree,GawrychowskiIIK18,RubinchikS20} and references therein).
In this paper, we treat a notion of palindromic structures
called \emph{minimal unique palindromic substring}~(\emph{MUPS})
that is introduced in~\cite{inoue2018algorithms}.
A palindromic substring $T[i.. j]$ of a string $T$ is called a MUPS of $T$
if $T[i..j]$ occurs exactly once in $T$ and $T[i+1..j-1]$ occurs at least twice in $T$.
MUPSs are utilized for solving
the \emph{shortest unique palindromic substring}~(\emph{SUPS})
problem proposed by Inoue et al.~\cite{inoue2018algorithms},
which is motivated by an application in molecular biology.
They showed that there are no more than $n$ MUPSs in any length-$n$ string,
and proposed an $O(n)$-time algorithm to compute all MUPSs of a given string of length $n$ over an integer alphabet of size $n^{O(1)}$.
After that, Watanabe et al.~\cite{WatanabeNIBT20} considered the problem of computing MUPSs in an \emph{run-length encoded}~(\emph{RLE}) string.
They showed that there are no more than $m$ MUPSs in a string whose RLE size is $m$.
Also, they proposed an $O(m\log\sigma_R)$-time and $O(m)$-space algorithm to compute all MUPSs of a string given in RLE,
where $\sigma_R$ is the number of distinct single-character runs in the RLE string.
Recently, Mieno et al.~\cite{Mieno2020EERTREE}
considered the problems of computing palindromic structures in the \emph{sliding window} model.
They showed that the set of MUPSs in a sliding window can be maintained in a total of $O(n\log\sigma_W)$ time and $O(D)$ space
while a window of size $D$ shifts over a string of length $n$ from the left-end to the right-end,
where $\sigma_W$ is the maximum number of distinct characters in the windows.
This result can be rephrased as follows:
The set of MUPSs in a string of length $D$ can be updated in
amortized $O(\log \sigma_W)$ time using $O(D)$ space
after deleting the first character or inserting a character to the right-end.

To the best of our knowledge, there is no efficient algorithm for updating the set of MUPSs after editing a character at \emph{any}
position so far.
Now, we consider the problem of updating the set of MUPSs in a string after substituting a character at any position.
Formally, we tackle the following problem:
Given a string $T$ of length $n$ over an integer alphabet of size $n^{O(1)}$ to preprocess,
and then given a query of single-character substitution.
Afterwards, we return the set of MUPSs of the edited string.
In this paper, we first show that the number $d$ of changes
of MUPSs after a single-character substitution is $O(\log n)$.
In addition, we present an algorithm that uses $O(n)$ time and space for preprocessing,
and updates the set of MUPSs in $O(\log\sigma + (\loglog n)^2 + d) \subset O(\log n)$ time.
We also propose a variant of the algorithm, which runs in optimal $O(1+d)$ time when the alphabet size is constant.

\subsubsection*{Related Work.}
There are some results for the problem of computing string regularities, including palindromes,
on dynamic strings~\cite{AmirCPR19,DBLP:journals/corr/abs-1906-09732,AmirBCK19,Charalampopoulos20}.
In this work, we consider the problem of computing MUPSs after a single edit operation as a first step toward a fully dynamic setting.
This line of research was initiated by Amir et al.~\cite{AmirCIPR17},
who tackled the problem of computing the longest common factor after one edit.
After that, other notions of string regularities are treated in a similar setting~\cite{UrabeNIBT18,AbedinH0T18,Funakoshi2021longestpal_afteredit}.
In particular, regarding palindromic structures,
Funakoshi et al.~\cite{Funakoshi2021longestpal_afteredit} proposed algorithms for computing the longest palindromic substring
after single-character or block-wise edit operations.

\subsubsection*{Paper Organization.}
  This paper is organized as follows.
  Section~\ref{sec:preliminaries} introduces necessary notations and algorithmic tools.
  Section~\ref{sec:num_of_changes} analyzes the changes of the set of MUPSs after a single-character substitution.
  We then present an algorithm for updating the set of MUPSs after a single-character substitution in Section~\ref{sec:algorithm}.
  Finally, we conclude in Section~\ref{sec:conclusions}. \section{Preliminaries} \label{sec:preliminaries}
\subsection{Notations}
\subsubsection*{Strings.}
Let $\Sigma$ be an \emph{alphabet} of size $\sigma$.
An element of $\Sigma$ is called a \emph{character}.
An element of $\Sigma^\ast$ is called a \emph{string}.
The length of a string $T$ is denoted by $|T|$.
The \emph{empty string} $\varepsilon$ is the string of length $0$.
For a string $T = xyz$, then $x, y$, and $z$ are called a \emph{prefix}, \emph{substring}, and \emph{suffix} of $T$, respectively.
They are called a \emph{proper prefix}, \emph{proper substring}, \emph{proper suffix} of $T$ if $x \ne T$, $y \ne T$, and $z \ne T$, respectively.
For each $1 \le i \le |T|$, $T[i]$ denotes the $i$-th character of $T$.
For each $1 \le i \le j \le |T|$, $T[i.. j]$ denotes the substring of $T$ starting at position $i$ and ending at position $j$.
For convenience, let $T[i'.. j'] = \varepsilon$ for any $i' > j'$.
A positive integer $p$ is said to be a \emph{period} of a string $T$ if $T[i] = T[i+p]$ for all $1 \le i \le |T|-p$.
For strings $X$ and $Y$, let $\lcp(X, Y)$ denotes the length of the \emph{longest common prefix}~(in short, lcp) of $X$ and $Y$, i.e.,
$\lcp(X, Y) = \max\{\ell \mid X[1.. \ell] = Y[1.. \ell]\}$.
For a string $T$ and two integers $1 \le i \le j \le |T|$, let $\lce_T(i, j)$ denotes
the length of the \emph{longest common extension}~(in short, lce) of $i$ and $j$ in $T$, i.e., $\lce_T(i, j) = \lcp(T[i..|T|], T[j..|T|])$.
For non-empty strings $T$ and $w$, $\beg_T(w)$ denotes the set of beginning positions of occurrences of $w$ in $T$.
Also, for a text position $i$ in $T$,
$\inbeg_{T, i}(w)$ denotes the set of beginning positions of occurrences of $w$ in $T$ where each occurrence covers position $i$.
Namely, $\beg_T(w) = \{b \mid T[b.. e] = w\}$ and $\inbeg_{T, i}(w) = \{b \mid T[b.. e] = w \textrm{ and } i \in [b, e]\}$.
Further, let $\xbeg_{T, i}(w) = \beg_T(w) \setminus \inbeg_{T, i}(w)$.
For convenience, $|\beg_T(\varepsilon)| = |\inbeg_{T, i}(\varepsilon)| = |\xbeg_{T, i}(\varepsilon)| = |T| + 1$ for any $T$ and $i$.
We say that $w$ is \emph{unique} in $T$ if $|\beg_T(w)| = 1$, and that $w$ is \emph{repeating} in $T$ if $|\beg_T(w)| \ge 2$.
Note that the empty string is repeating in any other string.
Since every unique substring $u = T[i..j]$ of $T$ occurs exactly once in $T$, we will sometimes identify $u$ with its corresponding interval $[i, j]$.
In what follows, we consider an arbitrarily fixed string $T$ of length $n \ge 1$ over an alphabet $\Sigma$ of size $\sigma = n^{O(1)}$.
\subsubsection*{Palindromes.}
For a string $w$, $w^R$ denotes the reversed string of $w$.
A string $w$ is called a \emph{palindrome} if $w = w^R$.
A palindrome $w$ is called an \emph{even-palindrome}~(resp. \emph{odd-palindrome}) if its length is even~(resp. odd).
For a palindrome $w$, its length-$\lfloor |w|/2 \rfloor$ prefix~(resp. length-$\lfloor |w|/2 \rfloor$ suffix)
is called the \emph{left arm}~(resp. \emph{right arm}) of $w$, and is denoted by $\larm_w$~(resp. $\rarm_w$).
Also, we call $\Larm_w = \larm_w\cdot s_w$~(resp. $\Rarm_w = s_w\cdot\rarm_w$) the \emph{extended left arm}~(resp. \emph{extended right arm}) of $w$ where $s_w$ is the character at the center of $w$ if $w$ is an odd-palindrome, and $s_w$ is empty otherwise.
Note that when $w$ is an even-palindrome, $\Rarm_w = \rarm_w$ and $\Larm_w = \larm_w$.
For a non-empty palindromic substring $w = T[i.. j]$ of a string $T$, the center of $w$ is $\frac{i+j}{2}$ and is denoted by $\centerpal(w)$.
A non-empty palindromic substring $T[i.. j]$ of a string $T$ is said to be \emph{maximal} if $i = 1$, $j = n$, or $T[i-1] \ne T[j+1]$.
For a non-empty palindromic substring $w = T[i.. j]$ of a string $T$ and a non-negative integer $\ell$,
$v = T[i-\ell.. j+\ell]$ is said to be an \emph{expansion} of $w$ if $1 \le i-\ell \le j+\ell \le n$ and $v$ is a palindrome.
Also, $T[i+\ell.. j-\ell]$ is said to be a \emph{contraction} of $w$.
Now we show periodic properties of palindromic suffixes of $T[1..i]$.
Let $\mathbf{S}_i = \{ s_1, \ldots, s_{g}\}$ be the set of lengths of palindromic suffixes of $T[1..i]$,
where $g$ is the number of palindromic suffixes of $T[1..i]$
and $s_{k-1} \leq s_k$ for $2 \leq k \leq g$.
Let $p_k$ be the progression difference for $s_k$,
i.e., $p_k = s_{k} - s_{k-1}$ for $2 \leq k \leq g$.
For convenience, let $p_1 = 1$.

Then, the following results are known:

\begin{lemma}[\cite{Apostolico1995parallel,GasieniecSWAT96,matsubara_tcs2009}]
  \label{lem:palindromic_suffixes}
  \hfill
  \begin{enumerate}
    \item[(A)] For any $1 \leq k < g$, $p_{k+1} \geq p_{k}$.
    \item[(B)] For any $1 < k < g$, if $p_{k+1} \neq p_{k}$, then $p_{k+1} \geq p_{k} + p_{k-1}$.
    \item[(C)] $\mathbf{S}_i$ can be represented by $O(\log i)$ arithmetic progressions,
      where each arithmetic progression is a tuple $\langle a, p, f \rangle$ representing the sequence $a, a+p, \ldots, a + (f-1)p$ of lengths of $f$ palindromic suffixes with common difference $p$.
    \item[(D)] The common difference $p$ is a smallest period of every palindromic suffixes of $T[1..i]$ whose length belongs to the arithmetic progression $\langle a, p, f \rangle$.
  \end{enumerate}
\end{lemma}

For each $1 \leq j \leq f$,
we will use the convention that $a(j) = a + (j-1)p$,
namely $a(j)$ denotes the $j$-th shortest element for $\langle a, p, f \rangle$.
For simplicity, let $Y = T[1..i]$ and $Z = T[i+1..n]$.
Let $\lcp(x,y)$ for strings $x$ and $y$ denote the length of the longest common prefix of $x$ and $y$.
Also, let $\Ext(a(j))$ denote the length of the maximal palindrome
that is obtained by extending $a(j)$ in $YZ$.

\begin{lemma}[\cite{matsubara_tcs2009}]\label{lem:batched_extension}
  For any $\langle a, p, f \rangle$ of palndromic suffixes of $T[1..i]$,
  there exist palindromes $u, v$ and a non-negative integer $q$,
  such that $(uv)^{f+q-1} u$ (resp. $(uv)^q u$) is
  the longest (resp. shortest) palindromic suffixes represented by $\langle a, p, f \rangle$
  with $|uv| = p$.
  Let $\alpha = \lcp((Y[1..|Y|-a(1)])^R, Z)$
  and $\beta = \lcp((Y[1..|Y|-a(f)])^R, Z)$.
  Then $\Ext(a(j)) = a(j) + 2\min \{\alpha, \beta + (f-j)p\}$.
  Further, if there exists $a(h) \in \langle a, p, f \rangle$ such that $a(h) + \alpha = a(f) + \beta$,
  then $\Ext(a(h)) = a(h) + 2\lcp((Y[1..|Y|-a(h)])^R, Z) \geq \Ext(a(j))$
  for any $j \neq h$.
\end{lemma}

By Lemmas~\ref{lem:palindromic_suffixes} and~\ref{lem:batched_extension}, we obtain the following corollary:
\begin{corollary} \label{cor:group_of_palsuf}
  For a position $i$, divide the set of palindromic suffixes of $T[1..i]$ into groups $G_1, G_2, \ldots, G_{m_i}$
  w.r.t their smallest periods.
  Then, $m_i \in O(\log i)$ holds.
  Also, for each group $G_k$, the following properties hold:
  \begin{enumerate}
    \item \label{item:diff_of_centers}
      The difference between centers of any two palindromes in $G_k$ is an integer power of $0.5p_k$, where $p_k$ is their smallest period.
    \item \label{item:period_of_mpal}
      For all maximal palindromes $e_1, \ldots, e_t$ in $T$ that are expansions of palindromes in $G_k$,
      excluding at most one, their smallest period is also $p_k$.
    \item \label{item:longest_candidate}
      Among $e_1, \ldots, e_t$, any palindrome is contained by the longest one or the second longest one.
  \end{enumerate}
\end{corollary}
We remark that symmetric arguments hold for palindromic prefixes as well.

A non-empty string $w$ is called a \emph{1-mismatch palindrome} if
there is exactly one mismatched position between
$w[1..\lfloor |w|/2 \rfloor]$ and $w[\lceil |w|/2 \rceil +1..|w|]^R$.
Informally, a 1-mismatch palindrome is a pseudo palindrome with a mismatch position between their arms.
As in the case of palindromes,
a 1-mismatch palindromic substring $T[i.. j]$ of a string $T$ is said to be \emph{maximal} if $i = 1$, $j = n$ or $T[i-1] \ne T[j+1]$.

A palindromic substring $T[i.. j]$ of a string $T$ is called a \emph{minimal unique palindromic substring} (\emph{MUPS}) of $T$ if $T[i..j]$ is unique in $T$ and $T[i+1..j-1]$ is repeating in $T$.
We denote by $\MUPS(T)$ the set of intervals corresponding to MUPSs of a string $T$.
A MUPS cannot be a substring of another palindrome with a different center.
Also, it is known that the number of MUPSs of $T$ is at most $n$,
and set $\MUPS(T)$ can be computed in $O(n)$ time for a given string $T$ over an integer alphabet~\cite{inoue2018algorithms}.
The following lemma states that the total sum of occurrences of strings which are extended arms of MUPSs is $O(n)$:
\begin{lemma}\label{lem:different_occ}
  The total sum of occurrences of the extended right arms of all MUPSs in a string $T$ is at most $2n$.
  Similarly, the total sum of occurrences of the extended left arms of all MUPSs in $T$ is at most $2n$.
\end{lemma}
\begin{proof}
  It suffices to prove the former statement for the extended \emph{right} arms since the latter can be proved symmetrically.
  Let $w_1$ and $w_2$ be distinct odd-length MUPSs of $T$ with $|w_1| \le |w_2|$.
  For the sake of contradiction, we assume that $T[j.. j+|\Rarm_{w_1}|-1] = \Rarm_{w_1}$ and $T[j.. j+|\Rarm_{w_2}|-1] = \Rarm_{w_2}$ for some position $j$ in $T$.
  Namely, $\Rarm_{w_1}$ is a prefix of $\Rarm_{w_2}$. Then, $\larm_{w_1}$ is a suffix of $\larm_{w_2}$ by palindromic symmetry.
  This means that $w_1$ is a substring of $w_2$.
  This contradicts that $w_2$ is a MUPS of $T$.
  Thus, all occurrences of the extended right arms of all odd-length MUPSs are different, i.e., the total number of the occurrences is at most $n$.
  Similarly, the total number of all occurrences of the right arms of all even-length MUPSs is also at most $n$.
\end{proof}
 \subsection{Tools}
This subsection lists some data structures used in our algorithm.
Our model of computation is a standard word RAM model with machine word size $\Omega(\log n)$.
\subsubsection*{Suffix Trees.}
The \emph{suffix tree} of $T$ is the compacted trie for all suffixes of $T$~\cite{Weiner1973suffixtree}.
We denote by $\STree(T)$ the suffix tree of $T$.
If a given string $T$ is over an integer alphabet of size $n^{O(1)}$,
$\STree(T)$ can be constructed in $O(n)$ time~\cite{Farach-Colton2000suffixtree}.
Not all substrings of $T$ correspond to nodes in $\STree(T)$.
However, the loci of such substrings can be made explicit in linear time:
\begin{lemma}[Corollary~8.1 in \cite{Kociumaka2020seedscomputation}] \label{lem:make_explicit_node}
  Given $m$ substrings of $T$, represented by intervals in $T$,
  we can compute the locus of each substring in $\STree(T)$ in $O(n + m)$ total time.
  Moreover, the loci of all the substrings in $\STree(T)$ can be made explicit in $O(n + m)$ extra time.
\end{lemma}
Also, this lemma implies the following corollary:
\begin{corollary}\label{cor:sort_substrings}
  Given $m$ substrings of $T$, represented by intervals in $T$, we can sort them in $O(n + m)$ time.
\end{corollary}
\subsubsection*{LCE Queries.}
An \emph{LCE query} on a string $T$ is, given two indices $i, j$ of $T$, to compute $\lce_T(i, j)$.
Using $\STree(T\$)$ enhanced with a lowest common ancestor data structure,
we can answer any LCE query on $T$ in constant time
where $\$$ is a special character with $\$ \not\in \Sigma$.
In the same way, we can compute the lcp value between any two suffixes of $T$ or $T^R$ in constant time
by using $\STree(T\$T^R\#)$ where $\#$ is another special character with $\# \not\in \Sigma$.
\subsubsection*{NCA Queries.}
A \emph{nearest colored ancestor query}~(NCA query) on a tree $\mathcal{T}$ with \emph{colored} nodes is,
given a query node $v$ and a color $C$, to compute the nearest ancestor $u$ of $v$ such that the color of $u$ is $C$.
Noticing that the notion of NCA is a generalization of well-known \emph{nearest marked ancestor}.
For NCA queries, we will use the following known results:
\begin{lemma}[\cite{Gawrychowski2018NCA}] \label{lem:nca_general}
  Given a tree $\mathcal{T}$ with colored nodes,
  a data structure of size $O(N)$ can be constructed in deterministic $O(N\log\log N)$ time or expected $O(N)$ time
  to answer any NCA query in $O(\log\log N)$ time, where $N$ is the number of nodes of $\mathcal{T}$.
\end{lemma}

\begin{lemma}[\cite{Bille2015GrammerCompressedString,Charalampopoulos2021Internal}] \label{lem:nca_log_colors}
  If the number of colors is $O(\log N)$,
  a data structure of size $O(N)$ can be constructed in $O(N)$ time
  to answer any NCA query in $O(1)$ time.
\end{lemma}
\subsubsection*{Eertrees.}
The \emph{eertree}~(a.k.a. palindromic tree) of $T$ is a pair of rooted edge-labeled trees $\Todd$ and $\Teven$
representing all distinct palindromes in $T$~\cite{Rubinchik2018eertree}.
The roots of $\Todd$ and $\Teven$ represent $\varepsilon$.
Each non-root node of $\Todd$~(resp. $\Teven$) represents an odd-palindrome~(resp. even-palindrome) which occurs in $T$.
Let $\pal(v)$ be the palindrome represented by a node $v$.
For the root $r_{\mathsf{odd}}$ of $\Todd$,
there is an edge $(r_{\mathsf{odd}}, u)$ labeled by $a \in \Sigma$ if there is a node $u$ with $\pal(u) = a$.
For any node $v$ in the eertree except for $r_{\mathsf{odd}}$,
there is an edge $(v, w)$ labeled by $a \in \Sigma$ if there is a node $w$ with $\pal(w) = a\cdot\pal(v)\cdot a$.
We denote by $\EERTREE(T)$ the eertree of $T$.
We will sometimes identify a node $u$ in $\EERTREE(T)$ with its corresponding palindrome $\pal(u)$.
Also, the path from a node $u$ to a node $v$ in $\EERTREE(T)$ is denoted by $\pal(u) \rightsquigarrow \pal(v)$.
If a given string $T$ is over an integer alphabet of size $n^{O(1)}$, $\EERTREE(T)$ can be constructed in $O(n)$ time~\cite{Rubinchik2018eertree}.
\subsubsection*{Path-Tree LCE Queries.}
A \emph{path-tree LCE query} is a generalized LCE query on a rooted edge-labeled tree $\mathcal{T}$~\cite{Bille2016LCE}:
Given three nodes $u$, $v$, and $w$ in $\mathcal{T}$ where $u$ is an ancestor of $v$,
to compute the lcp between the path-string from $u$ to $v$ and any path-string from $w$ to a descendant leaf.
The following result is known:
\begin{theorem}[Theorem 2 of \cite{Bille2016LCE}] \label{thm:pathtreeLCE}
  For a tree $\mathcal{T}$ with $N$ nodes,
  a data structure of size $O(N)$ can be constructed in $O(N)$ time
  to answer any path–tree LCE query in $O((\log\log N)^2)$ time.
\end{theorem}
We will use later path-tree LCE queries on the eertree of the input string.
\subsubsection*{Stabbing Queries.}
Let $\mathcal{I}$ be a set of $n$ intervals, each of which is a subinterval of the universe $U = [1, O(n)]$.
An \emph{interval stabbing query} on $\mathcal{I}$ is, given a query point $q \in U$,
to report all intervals $I \in \mathcal{I}$ such that $I$ is \emph{stabbed} by $q$, i.e., $q \in I$.
We can answer such a query in $O(1+k)$ time after $O(n)$-time preprocessing,
where $k$ is the number of intervals to report~\cite{Schmidt2009stabbing}.
 \section{Changes of MUPSs After Single Character Substitution} \label{sec:num_of_changes}
In the following, we fix
the original string $T$ of length $n$,
the text position $i$ in $T$ to be substituted, and
the string $T'$ after the substitution.
Namely, $T[i] \neq T'[i]$ and $T[j] = T'[j]$ for each $j$ with $1 \le j \le n$ and $j \neq i$.
This section analyzes the changes of the set of MUPSs when $T[i]$ is substituted by $T'[i]$.
For palindromes covering editing position $i$, Lemma~\ref{lem:same_pal_cannot_cover_i_after_edit} holds.
\begin{lemma}\label{lem:same_pal_cannot_cover_i_after_edit}
  For a palindrome $w$, if $\inbeg_{T, i}(w) \neq \emptyset$, then $\inbeg_{T', i}(w) = \emptyset$.
\end{lemma}
\begin{proof}
For the sake of contradiction, we assume that there is a palindrome $w$ with $\inbeg_{T, i}(w) \ne \emptyset$ and $\inbeg_{T', i}(w) \ne \emptyset$.
Let $c$~(resp. $c'$) be the center of an occurrence of $w$ in $T$~(resp. in $T'$) covering position $i$.
It is clear that $c \neq c'$ since $T[i]$ is substituted by another character $T'[i]$.
Also, it suffices to consider when $c < c'$ from the symmetry of $T$ and $T'$.
Let $d$~(resp. $d'$) be the distance between $c$ and $i$~(resp. $c'$ and $i$), i.e., $d = |i-c|$ and $d' = |i-c'|$.

There are the following two cases:
either (1)~$i \not\in [c,c']$ or (2)~$i \in [c,c']$.
See also Fig.~\ref{fig:same_pal} for illustration.
(1)~Now we consider the case when $c' < i$. Another case~($i < c$) can be treated similarly.
On the one hand, since $c$ and $c'$ are the centers of $w$, $T'[c'+d] = T[c+d] = T[i]$.
Further, $T'[c'-d] = T[c'+d]$ by palindromic symmetry, and hence, $T'[c'-d] = T[i]$.
On the other hand, again, since $c$ and $c'$ are the centers of $w$, $T[c-d'] = T'[c'-d']$.
Further, $T'[c'-d'] = T'[c'+d'] = T'[i]$ by palindromic symmetry, and hence, $T[c-d'] = T'[i]$.
Also, $T'[c-d'] = T[c-d'] = T'[i]$ since $c-d' \neq i$.
Since $c'-d = c-d'$ holds in this case, $T[i] = T'[c'-d] = T'[c-d'] = T'[i]$, a contradiction.

(2)~Similar to the first case, it can be seen that $T'[c'-d] = T[c-d] = T[i]$.
If $d = d'$, then $T[i] = T'[c'-d] = T'[c'-d'] = T'[i]$, a contradiction.
Hence $d \neq d'$ holds, and thus, $T'[c+d'] = T[c+d']$.
Also, $T[c+d'] = T'[c'+d'] = T'[c'-d'] = T'[i]$ holds.
Finally, since $c'-d = c+d'$ holds in this case, $T[i] = T'[c'-d] = T'[c+d'] = T'[i]$, a contradiction.
\end{proof}
\begin{figure}[tb]
\centerline{
  \includegraphics[scale=0.25]{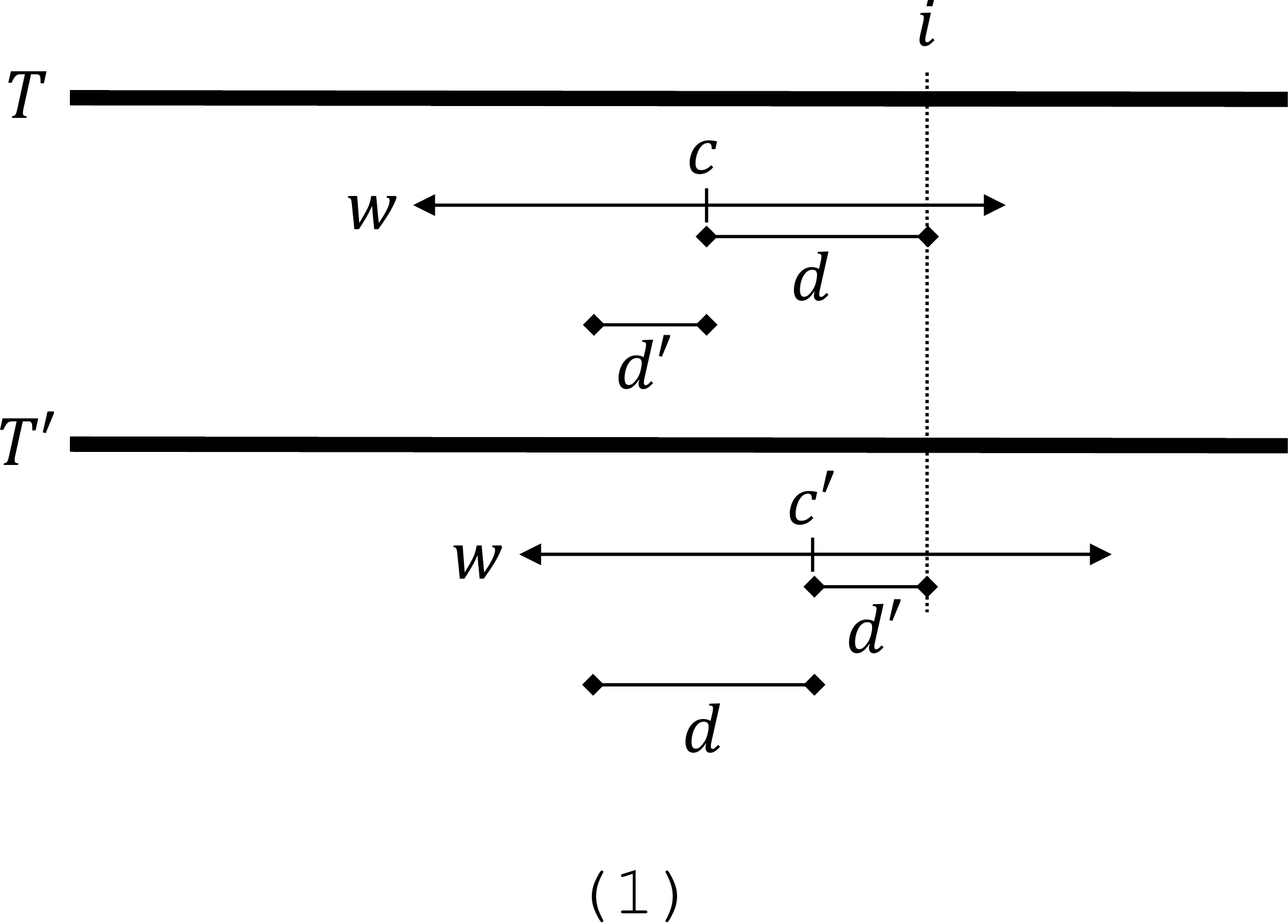}
  \hfill
  \includegraphics[scale=0.25]{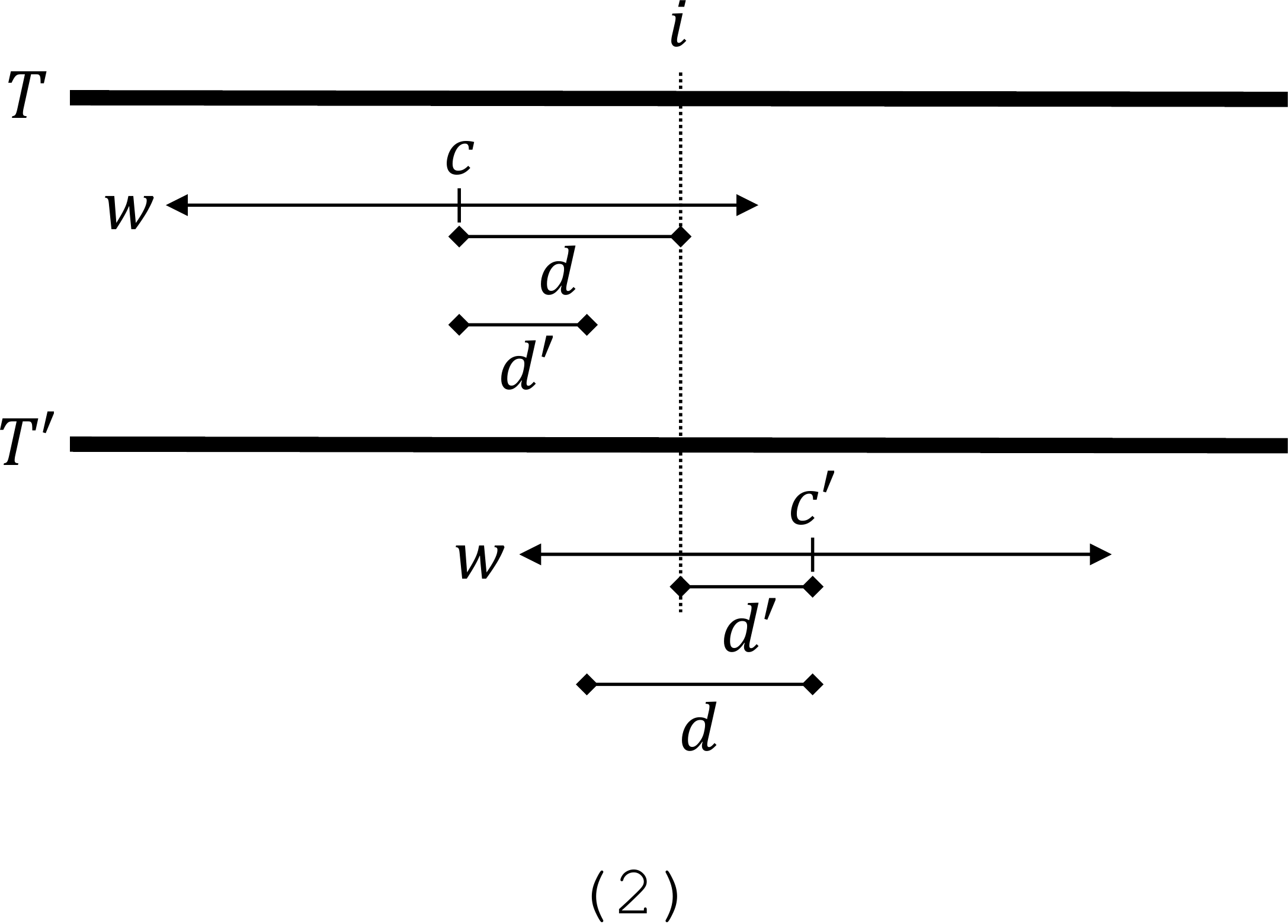}
}
\caption{Illustration for the two cases of Lemma~\ref{lem:same_pal_cannot_cover_i_after_edit}.
Note that this illustration is for the sake of contradiction.}
\label{fig:same_pal}
\end{figure}
For a position $i$, let $\mathcal{W}_i$ be the set of palindromes $w$ such that
$|\inbeg_{T, i}(w)| \ge 1$, $|\xbeg_{T, i}(w)|= 1$, and $w$ is minimal, i.e.,
$|\inbeg_{T, i}(v)| = 0$ or $|\xbeg_{T, i}(v)| \ge 2$ where $v = w[2..|w|-1]$.
This set $\mathcal{W}_i$ is useful for analyzing the number of changes of MUPSs in the proof of Theorem~\ref{thm:num_of_removed_MUPS}.
\begin{lemma} \label{lem:size_of_W}
  For any position $i$ in $T$, $|\mathcal{W}_i| \in O(\log n)$.
\end{lemma}
\begin{proof}
  First, by minimality of palindromes in $\mathcal{W}_i$, centers of palindromes in $\mathcal{W}_i$ are different from each other.
  Let $\mathcal{W}_i^L \subset \mathcal{W}_i$~(resp. $\mathcal{W}_i^R\subset \mathcal{W}_i$) be the set of palindromes in $\mathcal{W}_i$ whose center is at most $i$~(resp. at least $i$).

  Let us consider the size of $\mathcal{W}_i^L$.
  Every palindrome in $\mathcal{W}_i^L$ is an expansion of some palindromic suffix of $T[1.. i]$.
From Corollary~\ref{cor:group_of_palsuf},
  the set of palindromic suffixes of $T[1.. i]$
  can be divided into $m_i \in O(\log i)$ groups w.r.t. their smallest period.
  Let $G_1, G_2, \ldots, G_{m_i}$ be such groups for palindromic suffixes of $T[1.. i]$,
  and let $p_k$ be the smallest period corresponding to $G_k$.
  Also, for each $k$ with $1 \le k \le m_i$, let $E_k$ be the union set of all expansions of every palindrome in $G_k$.
  Let $H_k = \mathcal{W}_i^L \cap E_k$.
  Since $|\mathcal{W}_i^L| = |\bigcup_{k=1}^{m_i} H_k| = \sum_{k=1}^{m_i}|H_k|$ and $m_i \in O(\log n)$,
  it suffices to show that $|H_k|$ is at most a constant.

  For the sake of contradiction, we assume $|H_k|\ge 4$.
  From (\ref{item:period_of_mpal}) of Corollary~\ref{cor:group_of_palsuf},
  at least three palindromes in $H_k$ has the same smallest period $p_k$.
  Also, by (\ref{item:diff_of_centers}) of Corollary~\ref{cor:group_of_palsuf},
  the difference between centers of any two palindromes in $H_k$ is a power of $0.5p_k$.
  Therefore, at least two distinct palindromes in $H_k$ have a difference of power of $p_k$ in their center positions.
  Let $w_1$, $w_2$ be such palindromes, and assume that $|w_1| \le |w_2|$ w.l.o.g..
  Then, the string between the centers of $w_1$ and $w_2$ can be represented by $x^j$ for positive integer $j$
  and string $x$ of length $p_k$.
Since the smallest period of $w_1$ is $p_k$, its extended right arm $\Rarm_{w_1}$ can be written by $\Rarm_{w_1} = x^{j_1}x'_1$ where $j_1$
  is a non-negative integer and $x'_1$ is a proper prefix of $x$.
  Similarly, the extended right arm $\Rarm_{w_2}$ of $w_2$ can be written by $\Rarm_{w_2} = x^{j_2}x'_2$ where $j_2$
  is a non-negative integer and $x'_2$ is a proper prefix of $x$.
  See also Fig.~\ref{fig:size_of_W} for illustration.
If $|w_1| = |w_2|$, then this leads $j_1 = j_2$
  and $x'_1 = x'_2$, i.e., $w_1 = w_2$, a contradiction.
  If $|w_1| < |w_2|$, then $j_1 < j_2$ or $j_1 = j_2$
  and $|x'_1| < |x'_2|$.
  In both cases, $\Rarm_{w_1}$ is a proper prefix of $\Rarm_{w_2}$, i.e, $w_1$ is a contraction of $w_2$.
  This contradicts the minimality of $w_2$.
  Thus $|H_k| \le 3$ holds, and hence, we obtain $|\mathcal{W}_i^L| = \sum_{k=1}^{m_i}|H_k| \le 3m_i \in O(\log n)$.

  Similarly, the size of $\mathcal{W}_i^R$ is also $O(\log n)$. Therefore, $|\mathcal{W}_i| \in O(\log n)$.
  \end{proof}
  \begin{figure}[tb]
    \centerline{
      \includegraphics[scale=0.40]{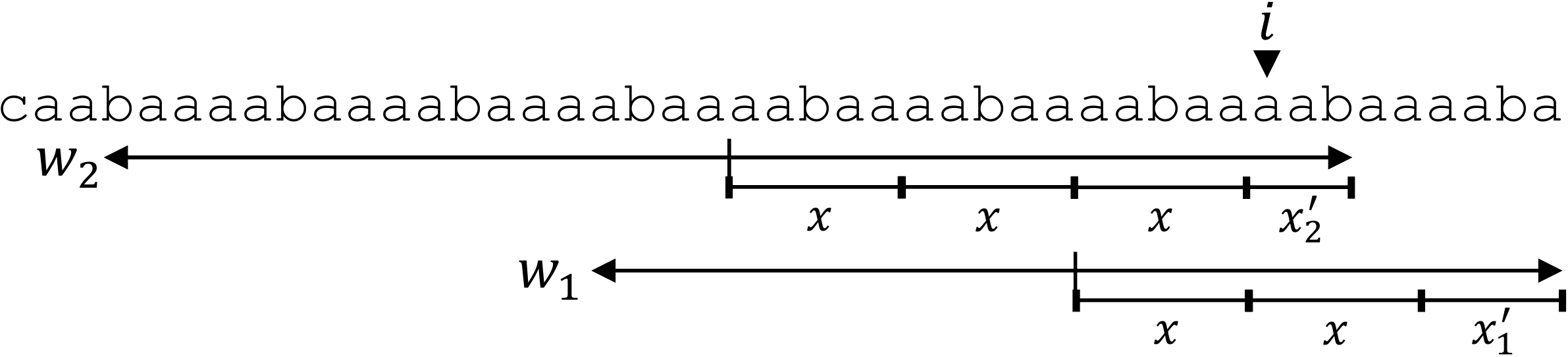}
    }
    \caption{Example for Lemma~\ref{lem:size_of_W}, where
    $w_1$ and $w_2$ have the same smallest period $5$ and the difference between centers of them is a power of $5$.
    Here $x=\mathtt{aabaa}$, $x'_1=\mathtt{aaba}$, and $x'_2=\mathtt{aab}$.
    }
      \label{fig:size_of_W}
\end{figure}
\begin{lemma}\label{lem:num_of_stabbed_MUPS}
  For each position $i$ in $T$, the number of MUPSs covering $i$ is $O(\log n)$.
\end{lemma}
\begin{proof}
  By symmetry, it suffices to show that
  the number of MUPSs covering $i$ and centered \emph{before} $i$ is $O(\log n)$.
  Each of such MUPSs is an expansion of some palindromic suffix of $T[1..i]$.
  Thus, similar to the proof of Lemma~\ref{lem:size_of_W},
  we consider dividing the set of palindromic suffixes of $T[1..i]$
  into $m_i \in O(\log i)$ groups, $G_1, G_2, \ldots, G_{m_i}$ w.r.t. their smallest periods.
In the following, we consider MUPSs that are expansions of palindromes in an arbitrary group $G_k$.
  We show that the number of such MUPSs is at most two by contradiction.
  We assume the contrary, i.e., there are three MUPSs that are expansions of palindromes in $G_k$.
  By (\ref{item:longest_candidate}) of Corollary~\ref{cor:group_of_palsuf},
  at least one of the three MUPSs is a substring of an expansion of a palindrome in $G_k$ with a different center.
  This contradicts that a MUPS cannot be a substring of another palindrome with a different center.
Thus, the number of MUPSs that are expansions of palindromes in $G_k$ is at most two, and we finish the proof.
\end{proof}
By using Lemmas~\ref{lem:same_pal_cannot_cover_i_after_edit}, \ref{lem:size_of_W}, and \ref{lem:num_of_stabbed_MUPS},
we show the following theorem:
\begin{theorem} \label{thm:num_of_removed_MUPS}
  $|\MUPS(T) \bigtriangleup \MUPS(T')| \in O(\log n)$ always holds.
\end{theorem}
\begin{proof}
In the following, we consider the number of MUPSs to be removed.

First, at most one interval can be a MUPS of $T$ centered at $i$.
Also, any other interval in $\MUPS(T)$ covering position $i$ cannot be an element of $\MUPS(T')$
since its corresponding string in $T'$ is no longer a palindrome.
By Lemma~\ref{lem:num_of_stabbed_MUPS}, the number of such MUPSs is $O(\log n)$.

Next, let us consider MUPSs not covering position $i$.
When a MUPS $w$ of $T$ not covering $i$ is no longer a MUPS of $T'$, then
either (A) $w$ is repeating in $T'$ or (B) $w$ is unique in $T'$ but is not minimal.

Let $w_1$ be a MUPS of the case (A).
Since $w_1$ does not cover $i$, is unique in $T$, and is repeating in $T'$,
$|\inbeg_{T',i}(w_1)| \ge 1$ and $|\xbeg_{T',i}(w_1)| = 1$.
Let $v_1$ be the minimal contraction of $w_1$ such that
$|\inbeg_{T',i}(v_1)| \ge 1$ and $|\xbeg_{T',i}(v_1)| = 1$.
Contrary, $w_1$ is the only MUPS of the case (A) which is an expansion of $v_1$ since $|\xbeg_{T',i}(v_1)| = 1$.
Namely, there is a one-to-one relation between $w_1$ and $v_1$.
By Lemma~\ref{lem:size_of_W}, the number of palindromes that satisfy the above conditions of $v_1$ is $O(\log n)$.
Thus, the number of MUPSs of the case (A) is also $O(\log n)$.

Let $w_2$ be a MUPS of the case (B).
In $T'$, $w_2$ covers some MUPS as a proper substring since it is not a MUPS and is unique in $T'$.
Let $v_2$ be the MUPS of $T'$, which is a proper substring of $w_2$.
While $v_2$ is unique in $T'$, it is repeating in $T$ since $w_2$ is a MUPS of $T$.
Namely, $|\inbeg_{T,i}(v_2)| \ge 1$ and $|\xbeg_{T,i}(v_2)| = 1$ hold.
Also, $v_2$ is actually minimal: Let $u_2 = v_2[2.. |v_2|-1]$.
If we assume that $|\inbeg_{T,i}(u_2)| \ge 1$ and $|\xbeg_{T,i}(u_2)| = 1$,
then $u_2$ becomes unique in $T'$, and this contradicts that $v_2$ is a MUPS of $T'$.
Furthermore, similar to the above discussions,
there is a one-to-one relation between $w_2$ and $v_2$.
Again by Lemma~\ref{lem:size_of_W}, the number of palindromes that satisfy the above conditions of $v_2$ is $O(\log n)$.
Thus, the number of MUPSs of the case (B) is also $O(\log n)$.

Therefore, $|\MUPS(T)\setminus\MUPS(T')| \in O(\log n)$ holds.
Also, $|\MUPS(T')\setminus\MUPS(T)| \in O(\log n)$ holds by symmetry.

To summarize,
$|\MUPS(T) \bigtriangleup\MUPS(T')| = |\MUPS(T)\setminus\MUPS(T') \cup \MUPS(T')\setminus\MUPS(T)| = |\MUPS(T)\setminus\MUPS(T')| + |\MUPS(T')\setminus\MUPS(T)| \in O(\log n)$.
\end{proof}
 \section{Algorithms for Updating Set of MUPSs} \label{sec:algorithm}
In this section, we propose an algorithm for updating the set of MUPSs when a single-character in the original string is substituted by another character.
We denote by $\sub(i, s)$ the substitution query, that is, to substitute $T[i]$ by another character $s$.
First, we define a sub-problem that will be used in our algorithm:
\begin{problem} \label{prob:lpp_centered_at_i}
Given a substitution query $\sub(i, s)$ on $T$, compute the longest odd-palindromic substring $v$ of $T'$
such that $\centerpal(v) = i$ and $v$ occurs in $T$ if it exists.
Also, if such $v$ exists, determine whether $v$ is unique in $T$ or not.
Furthermore, if $v$ is unique in $T$, compute the contraction of $v$ that is a MUPS of $T$.
\end{problem}
We show the following lemma:
\begin{lemma} \label{lem:lpp_centered_at_i}
  After $O(n)$-time preprocessing, we can answer Problem~\ref{prob:lpp_centered_at_i} in $O(\delta(n,\sigma) + (\log\log n)^2)$ time
  where $\delta(n,\sigma)$ denote the time to retrieve any child of the root of the odd-tree of $\EERTREE(T)$.
\end{lemma}
\begin{proof}
In the preprocessing, we construct $\EERTREE(T)$ and apply the preprocessing for the path-tree LCE queries
to the odd-tree $\Todd$ of $\EERTREE(T)$.
Also, we mark the nodes in $\EERTREE(T)$ that correspond to MUPSs of $T$
and apply the preprocessing for the nearest marked ancestor (NMA) queries.
The preprocessing time is $O(n)$.

Given a query $\sub(i, s)$, we query the path-tree LCE between
path $T[i] \rightsquigarrow \larm_wT[i]\rarm_w$
and tree rooted at $s$ on $\Todd$
where $w$ is the maximal palindrome in $T'$ centered at $i$.
Let $\ell_w$ be the depth of the LCE nodes.
Then the contraction $v$ of $w$ with $|\Rarm_v| = \ell_w$ occurs in $T$.
Also, $v$ is the longest since any other contraction $u$ of $w$ longer than $v$ does not occur in $T$.
Further, we can determine whether $v$ is unique in $T$ or not by checking the existence of a mark on
path $s \rightsquigarrow v$.
It can be done by querying NMA, and the MUPS of $T$ contained in $v$ can be computed simultaneously, if $v$ is unique in $T$.

We can compute the value $\ell_w$ in $\delta(n,\sigma)$ time for searching for the node $s$ in $\Todd$,
plus $O((\log\log n)^2)$ time for the path-tree LCE query.
\end{proof}
When $\sigma \in O(n)$, we can easily achieve $\delta(n, \sigma) \in O(1)$ with linear space, by using an array of size $\sigma$.
Otherwise, we achieve $\delta(n, \sigma) \in O(\log\sigma)$ for a general ordered alphabet by using a binary search tree.

In the rest of this paper, we propose an algorithm to compute the changes in MUPSs after a single-character substitution.
We compute MUPSs to be removed and added separately.
We show how to compute all MUPSs to be removed in Section~\ref{subsec:removedMUPS}.
Also, we show how to compute all MUPSs to be added in Section~\ref{subsec:addedMUPS}.
In Section~\ref{subsec:using_NCA}, we introduce another solution for Problem~\ref{prob:lpp_centered_at_i}.
Our strategy is basically to pre-compute changes in MUPSs for some queries as much as possible within linear time.
The other changes will be detected on the fly by using some data structures.
\subsection{Computing MUPSs to be Removed} \label{subsec:removedMUPS}

We categorize MUPSs to be removed into three types as follows:
\begin{description}
  \item[R1)] A MUPS of $T$ that covers $i$.
  \item[R2)] A MUPS of $T$ that does not cover $i$ and is repeating in $T'$.
  \item[R3)] A MUPS of $T$ that does not cover $i$ and is unique but not minimal in $T'$.
\end{description}
See Fig.~\ref{fig:types_of_removed_MUPS} for illustration.
\begin{figure}[tb]
  \centering
  \includegraphics[width=0.6\linewidth]{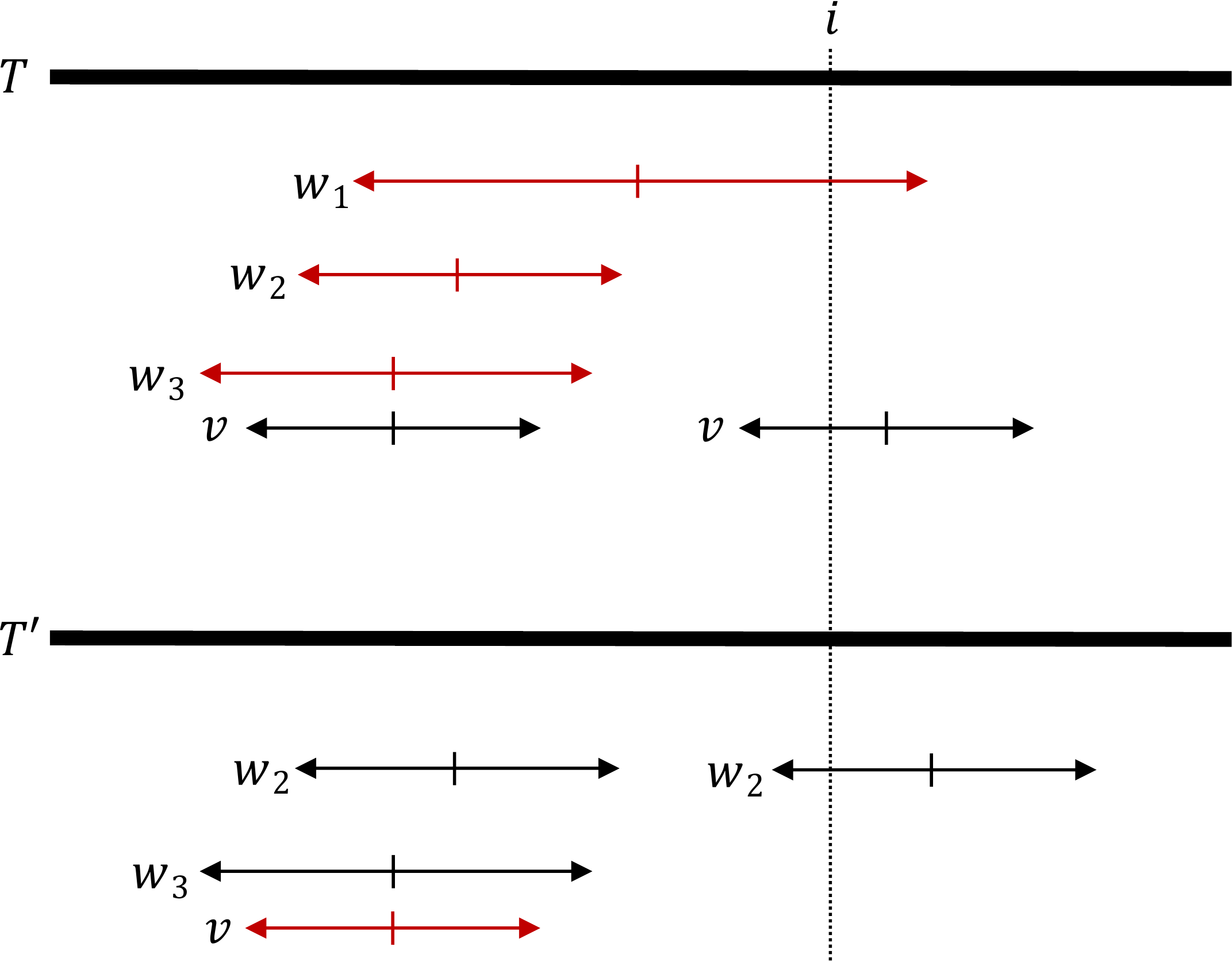}
  \caption{
    Illustration for three types of MUPSs to be removed.
    The red arrows represent MUPSs.
    $w_1$, $w_2$, and $w_3$ are MUPSs of Type R1, Type R2, and Type R3 in $T$, respectively.
    Also, $v$ is the MUPS of $T'$ that is a contraction of $w_3$.
    It is not unique in $T$, but is unique in $T'$.
  }\label{fig:types_of_removed_MUPS}
\end{figure}
In the following, we describe how to compute all MUPSs for each type separately.
\subsubsection*{Type R1.}
All MUPSs covering editing position $i$ are always removed.
Thus, we can detect them in $O(1+\alpha_{\rem})$ time after a simple linear time preprocessing (e.g., using stabbing queries),
where $\alpha_{\rem}$ is the number of MUPSs of Type R1.
\subsubsection*{Type R2.}
Before describing our algorithm, we give a few observations about MUPSs of Type R2.
Let $w$ be a MUPS of Type R2.
Since $w$ is unique in $T$ and is repeating in $T'$, $|\inbeg_{T', i}(w)| \ge 1$.
When $w$ occurs in $T'$ centered at editing position $i$,
we retrieve such $w$ by applying Problem~\ref{prob:lpp_centered_at_i}.
If it is not the case, we can utilize the following observations:
Consider the starting position $j$ of an occurrence of $w$ in $T'$ such that $T'[j.. j+|w|-1] = w$ and $i \in [j, j+|w|-1]$.
If position $i$ is covered in the right arm of $T'[j.. j+|w|-1]$, then $\Larm_w$ occurs at position $j$ in both $T$ and $T'$.
Further, the Hamming distance between $T[j+|\Larm_w|.. j+|w|-1]$ and $w[|\Larm_w|+1.. |w|] = \rarm_w$ equals $1$.
Namely, for each occurrence at position $k$ of string $\Larm_w$ in $T$,
$w$ can occur at $k$ in $T'$ only if the Hamming distance between $T[k+|\Larm_w|.. k+|w|-1]$ and $w[|\Larm_w|+1.. |w|]$ equals $1$.
In other words, if the Hamming distance is greater than $1$, $w$ cannot occur at $k$ in $T'$.
See also Fig.~\ref{fig:type_R2}.
\begin{figure}[tb]
  \centerline{
    \includegraphics[scale=0.45]{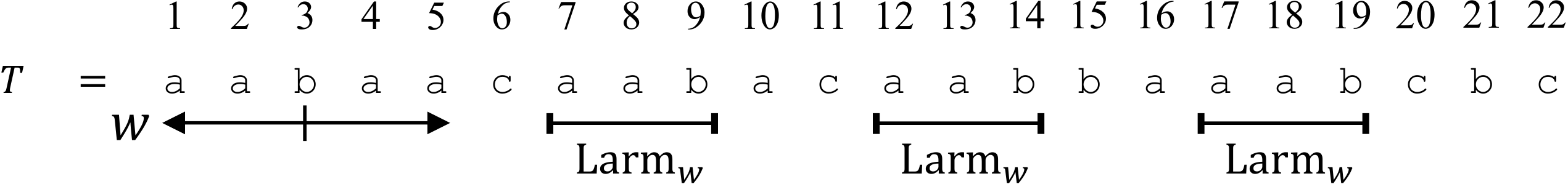}
  }
  \caption{Example for observations about Type R2,
    where $T = \mathtt{aabaacaabacaabbaaabcbc}$ and $w = \mathtt{aabaa}$.
    $\Larm_w = \mathtt{aab}$ occurs at position $7$, $12$, and $17$ excluding the occurrence of $w$.
    Since the Hamming distance between $T[10..11]$ and $w[4..5]$ equals 1, $w$ occurs at position $7$ when $T[11]$ is substituted by $\mathtt{a}$.
    Also, $w$ occurs at position $12$ when $T[15]$ is substituted by $\mathtt{a}$.
    Conversely, $w$ cannot occur at $17$ after any single-character substitution
    since the Hamming distance between $T[20..21]$ and $w[4..5]$ equals 2.
  }
    \label{fig:type_R2}
\end{figure}
The strategy of the algorithm for MUPSs of Type R2 is following:
A MUPS of Type R2 having occurrence centered at editing position $i$ in $T'$ is found by applying Problem~\ref{prob:lpp_centered_at_i}.
For the remaining MUPSs of Type R2, we precompute them by using the above observations.

\paragraph*{Preprocessing.}

In the preprocessing phase,
we first apply the $O(n)$-time preprocessing of Lemma~\ref{lem:lpp_centered_at_i} for Problem~\ref{prob:lpp_centered_at_i}.
Next, we initialize the set $\mathcal{A}_{R2} = \emptyset$.
The set $\mathcal{A}_{R2}$ will become an \emph{index} of MUPSs of Type R2 when the preprocessing is finished.
For each MUPS $w = T[b.. e]$ of $T$, we process the followings:
For the beginning position $j \ne b$ of each occurrence of $\Larm_w$ in $T$,
we first compute the lcp value between $T[j+|\Larm_w|..|T|]\$$ and $\rarm_w$ with allowing one mismatch.
Note that $T[j+|\Larm_w|..|T|]\$$ must have at least one mismatch with $\rarm_w$,
since $T[j.. j+|\Larm_w|-1] = \Larm_w$, $j \ne b$, and $T[b.. e]$ is unique in $T$.
If there are two mismatch positions between them, do nothing for this occurrence since $w$ cannot occur at $j$ after any substitution.
We can check this by querying LCE at most twice.
Otherwise, let $q = j+|\Larm_w|-1+d$ be the mismatched position in $T$.
When the $q$-th character of $T$ is substituted by
the character $\rarm_w[d]$, $w = T[b.. e]$ occurs at $j \ne b$, i.e., it is a MUPS of Type R2.
So we add MUPS $w = T[b.. e]$ into $\mathcal{A}_{R2}$ with the pair of index and character $(q, \rarm_w[d])$ as the key.
In addition, symmetrically, we update $\mathcal{A}_{R2}$ for each occurrence of $\Rarm_w$ in $T$.
After finishing the above processes for every MUPS of $T$, we sort the elements of $\mathcal{A}_{R2}$ by radix sort on the keys.
If there are multiple identical elements with the same key, we unify them into a single element.
Also, if there are multiple elements with the same key, we store them in a linear list.
By Lemma~\ref{lem:different_occ}, the total number of occurrences of arms of MUPSs is $O(n)$,
and hence, the total preprocessing time is $O(n)$.

\paragraph*{Query.}
Given a query $\sub(i, s)$,
we query Problem~\ref{prob:lpp_centered_at_i} with the same pair $(i, s)$ as the input.
Then, we complete checking whether there exists a MUPS of Type R2 centered at $i$.
Next, consider the existence of the remaining MUPSs of Type R2.
First, an element in $\mathcal{A}_{R2}$ corresponding to the key $(i, s)$
can be detected in $O(\log\sigma_i)$ time by using random access on indices and binary search on characters,
where $\sigma_i$ is the number of characters $s_i$ such that the key $(i, s_i)$ exists in $\mathcal{A}_{R2}$.
After that, we can enumerate all the other elements with the key by scanning the corresponding linear list.
Thus, the total query time is $O(\delta(n,\sigma) + (\log\log n)^2 +\log\sigma_i + \beta_{\rem})$ where $\beta_{\rem}$ is the number of MUPSs of Type R2.
Finally, we show $\sigma_i \in O(\min\{\sigma, \log n\})$.
Let us consider palindromes in $T'$ whose right arm covers position $i$.
Those whose left arms cover $i$ can be treated similarly.
Any palindrome in $T'$ whose right arm covers $i$ is an expansion of some maximal palindrome in $T$ ending at $i-1$.
It is known that the number of possible characters immediately preceding such maximal palindromes is $O(\log n)$~\cite{Funakoshi2021longestpal_afteredit}.
Therefore, $\sigma_i \in O(\log n)$ holds, and thus, the query time is
$O(\delta(n,\sigma) + (\log\log n)^2 + \beta_{\rem})$.
\subsubsection*{Type R3.}
Let $w = T[b.. e]$ be a MUPS of $T$ and let $v = T[b+1..e-1]$.
Further let $T[b_{l1}.. e_{l1}]$ and $T[b_{r1}.. e_{r1}]$ be the leftmost and the rightmost occurrence of $v$ in $T$ except for $T[b+1.. e-1]$.
We define interval
$\rho_w = \{ k \mid k \not\in [b+1, e-1] \text{ and } k \in [b_{r1}, e_{l1}]\}$.
Note that $\rho_w$ can be empty.
See also Fig.~\ref{fig:type_R3} for illustration.
\begin{figure}[tb]
  \centerline{
    \includegraphics[scale=0.35]{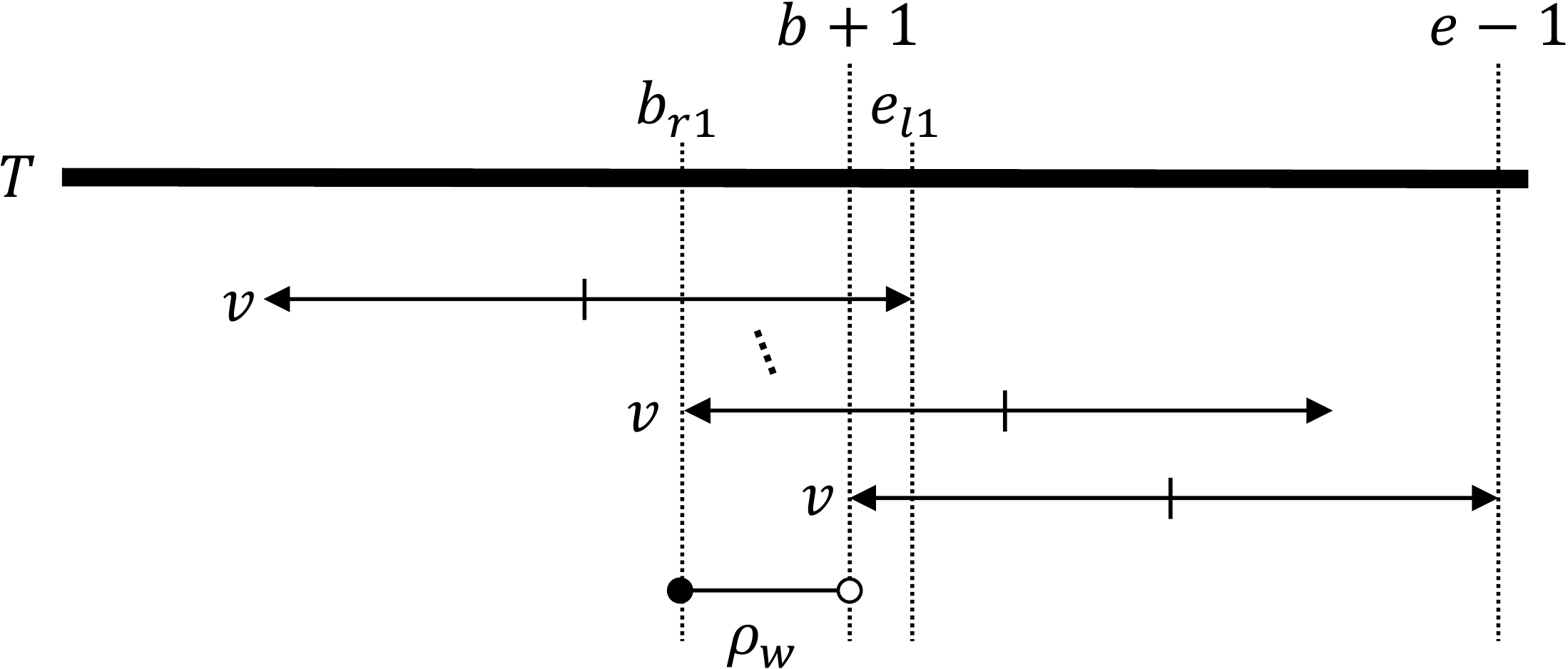}
  }
  \caption{Illustration for $\rho_w$ of Type R3.
  The top arrow (resp. the middle arrow) represents the leftmost (resp. rightmost) occurrence of $v$ except for $T[b+1..e-1]$.
  Also, the bottom arrow represents $T[b+1..e-1]$.
  In this case, $\rho_w = [b_{r1},b]$.
  }
    \label{fig:type_R3}
\end{figure}
If the editing position $i$ is in $\rho_w$, then the only occurrence of $v$ in $T'$ is $T'[b+1.. e-1]$, i.e., $v$ is unique in $T'$.
Thus, $w$ is a removed MUPS of Type R3.
Contrary, if $i \not\in [b, e]$ and $i \not\in \rho_w$, there are at least two occurrences of $v$ in $T'$, i.e., $w$ cannot be a MUPS of Type R3.

\paragraph*{Preprocessing.}
First, we compute the set of intervals $\mathcal{R} = \{\rho_w \mid w \text{ is a MUPS of }T \}$.
$\mathcal{R}$ can be computed by traversing over
the suffix tree of $T$ enhanced with additional explicit nodes, each of which represents a substring $T[b+1..e-1]$ for each MUPS $T[b..e]$ of $T$.
Also, we apply the preprocessing for stabbing queries to $\mathcal{R}$.
The total time for preprocessing is $O(n)$.

\paragraph*{Query.}
Given a query $\sub(i, s)$, compute all intervals in $\mathcal{R}$ stabbed by position $i$ by answering a stabbing query.
They correspond to MUPSs of Type R3.
The query time is $O(1 + \gamma_{\rem})$, where $\gamma_{\rem}$ is the number of MUPSs of Type R3.
\\
\\
\noindent
To summarize, we can compute all MUPSs to be removed after a single-character substitution
in $O(\delta(n,\sigma) + (\log\log n)^2 + \alpha_{\rem} + \beta_{\rem} + \gamma_{\rem})$ time.
 \subsection{Computing MUPSs to be Added} \label{subsec:addedMUPS}
Next, we propose an algorithm to detect MUPSs to be added after a substitution.
As in Section~\ref{subsec:removedMUPS}, we categorize MUPSs to be added into three types:
\begin{description}
  \item[A1)] A MUPS of $T'$ that covers $i$.
  \item[A2)] A MUPS of $T'$ that does not cover $i$ and is repeating in $T$.
  \item[A3)] A MUPS of $T'$ that does not cover $i$ and is unique but not minimal in $T$.
\end{description}
Furthermore, we categorize MUPSs of Type A1 into two sub-types:
\begin{description}
  \item[A1-1)] A MUPS of $T'$ that covers position $i$ in its arm.
  \item[A1-2)] A MUPS of $T'$ centered at editing position $i$.
\end{description}
See Fig.~\ref{fig:types_of_added_MUPS} for illustration.
\begin{figure}[tb]
  \centering
  \includegraphics[width=0.6\linewidth]{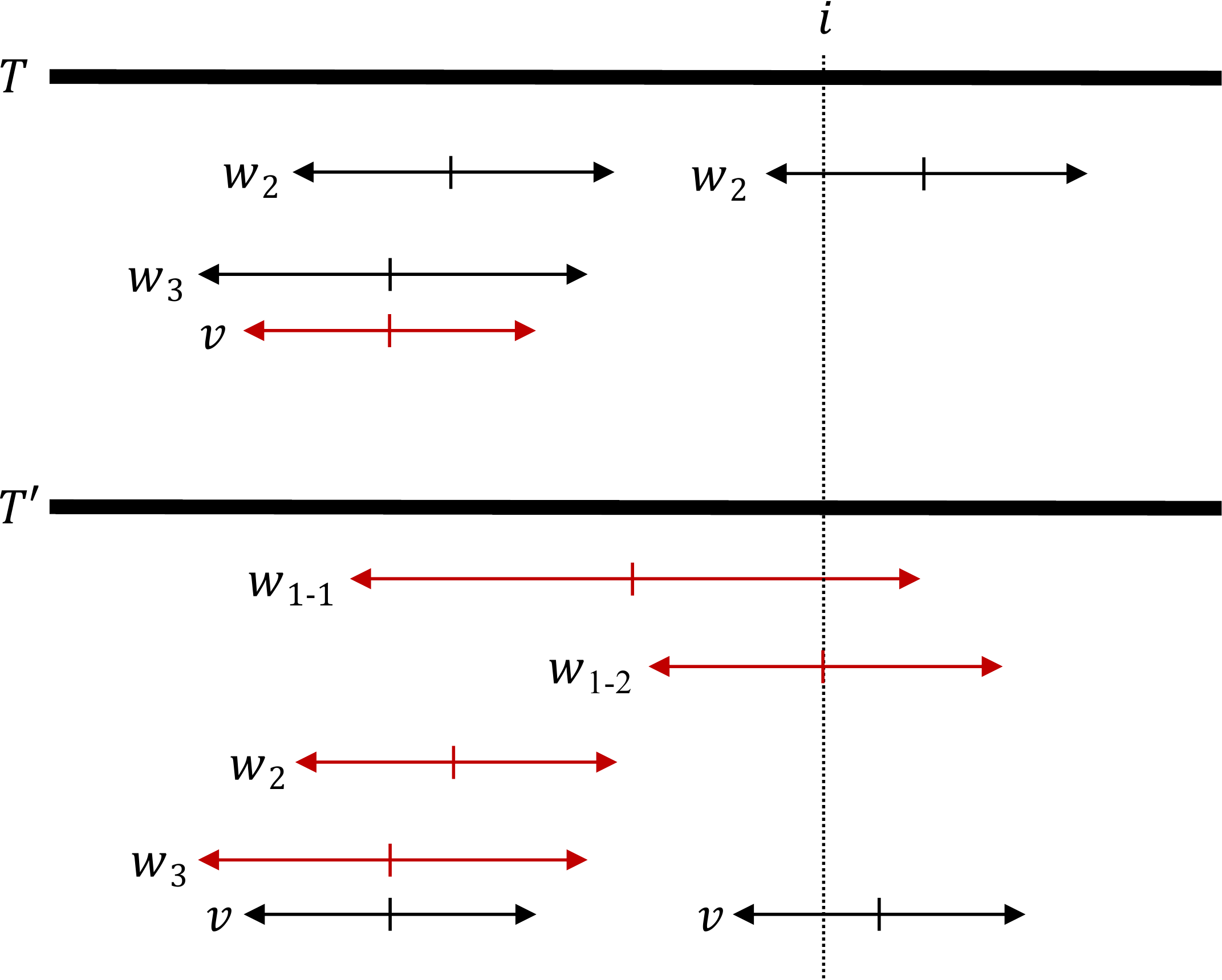}
  \caption{
    Illustration for four types of MUPSs to be added.
The red arrows represent MUPSs.
    $w_{1 \mathchar`- 1}$, $w_{1 \mathchar`- 2}$, $w_2$, and $w_3$ are MUPSs of Type A1-1, Type A1-2, Type A2, and Type A3 in $T'$, respectively.
    Also, $v$ is the MUPS of $T$ that is a contraction of $w_3$.
    It is unique in $T$, but is not unique in $T'$.
  }\label{fig:types_of_added_MUPS}
\end{figure}

\subsubsection*{Type A1-1.}
A MUPS of Type A1-1 is a contraction of some maximal palindrome in $T'$ covering editing position $i$ in its arm.
Further, such a maximal palindrome in $T'$ corresponds to some 1-mismatch maximal palindrome in $T$, which covers $i$ as a mismatch position.
Thus, we preprocess for arms of each 1-mismatch maximal palindrome in $T$.
For MUPSs of Type A1, we utilize the following observation:
\begin{observation} \label{obs:unique_covering_i}
  For any palindrome $v$ covering position $i$ in $T'$,
  $v$ is unique in $T'$ if and only if $|\inbeg_{T', i}(v)| = 1$ and $|\xbeg_{T', i}(v)| = 0$.
\end{observation}
\paragraph*{Preprocessing.}

In the preprocessing phase,
we first consider sorting extended arms of 1-mismatch maximal palindromes in $T$.
Let $\EARM$ be the multiset
of strings that consists of the extended right arms and the reverse of the extended left arms of all 1-mismatch maximal palindromes in $T$.
Note that each string in $\EARM$ can be represented in constant space since it is a substring of $T$ or $T^R$.
Let $\MA'$ be a lexicographically sorted array of all elements in $\EARM$.
Here, the order between the same strings can be arbitrary.
Also, for each string in $\EARM$, we consider a quadruple of the form $(\mathit{par}, \mathit{pos}, \mathit{chr}, \mathit{rnk})$ where
$\mathit{par} \in \{\odd, \even\}$ represents the parity of the length of the corresponded 1-mismatch maximal palindrome,
$\mathit{pos}$ is the mismatched position on the \emph{opposite} arm,
$\mathit{chr}$ is the mismatched character on the extended arm, and
$\mathit{rnk}$ is the rank of the extended arm in $\MA'$.
Let $\MA$ be a radix sorted array of these quadruples.
It can be seen that for each triple $(p, i, s)$ of parity $p$, mismatched position $i$, and mismatched character $s$,
all elements corresponding to the triple are stored continuously in $\MA$.
We denote by $\MA_{p,i,s}$ the subarray of $\MA$ consists of such elements.
In other words, $\MA_{p,i,s}$ is a sorted array of extended arms of
maximal palindromes of parity $p$
covering position $i$ in $T'$
when the $i$-th character of $T$ is substituted by $s$.

Now let us focus on odd-palindromes. Even-palindromes can be treated similarly.
We construct the suffix tree of $T$ and make the loci of strings in $\EARM$ explicit.
We also make the loci of the extended right arm of every odd-palindrome in $T$ explicit.
Simultaneously, we \emph{mark} the nodes corresponding to the extended right arms and apply the preprocessing for the nearest marked ancestor~(NMA) queries to the marked tree.
We denote the tree by $\mathcal{ST}_{\mathit{odd}}$.
Next, we initialize the set $\mathcal{A}_{A1,1} = \emptyset$.
The set $\mathcal{A}_{A1,1}$ will become an index of MUPSs of Type A1-1 when the preprocessing is finished.
For each non-empty $\MA_{\odd,i,s}$ and for each string $w$ in $\MA_{\odd,i,s}$, we do the followings:
Let $x_w$ be the odd-palindrome whose extended right arm is $w$ when $T[i]$ is substituted by $s$.
Let $u$ and $v$ are the preceding and the succeeding string of $w$ in $\MA_{\odd, i, s}$ (if such palindromes do not exist, they are empty).
Further let $\ell_w = \max\{\lcp(u, w), \lcp(w, v)\}$.
When $T[i]$ is substituted by $s$,
any contraction $y$ of $x_w$ such that $y$ covers position $i$ and the arm-length of $y$ is at least $\ell_w$
has only one occurrence which covers position $i$ in $T'$, i.e.,$|\inbeg_{T', i}(y)| = 1$.
Next, we query the NMA for the node corresponding to $w$
on $\mathcal{ST}_\mathit{odd}$.
Let $\ell'_w$ be the length of the extended right arm obtained by the NMA query.
When $T[i]$ is substituted by $s$,
any contraction $y'$ of $x_w$ such that the arm-length of $y'$ is at least $\ell'_w$, has no occurrences which do not cover position $i$ in $T'$,
i.e., $|\xbeg_{T', i}(y')| = 0$.
Thus, by Observation~\ref{obs:unique_covering_i}, the contraction $y^\star$ of $x_w$ of arm-length $\max\{\ell_w, \ell'_w\}$ is a MUPS of Type A1-1
for the query $\sub(i,s)$, if such $y^\star$ exists.
In such a case, we store the information about $y^\star$~(i.e., its center and radius) into $\mathcal{A}_{A1,1}$ using $(\odd, i, s)$ as the key.
After finishing the above preprocessing for all strings in $\MA$, we sort all elements in $\mathcal{A}_{A1,1}$ by their keys.

Since each element in $\EARM$ is a substring of $T\$T^R\#$,
they can be sorted in $O(n+|\EARM|) = O(n)$ time by Corollary~\ref{cor:sort_substrings}.
Namely, $\MA'$ can be computed in linear time, and thus $\MA$ too.
By Lemma~\ref{lem:make_explicit_node}, tree $\mathcal{ST}_\mathit{odd}$ can be constructed in $O(n)$ time.
Also, we can answer each NMA query and LCP query in constant time after $O(n)$ time preprocessing.
Hence, the total preprocessing time is $O(n)$.

\paragraph*{Query.}
Given a query $\sub(i, s)$, we compute all MUPSs of Type A1-1
by searching for elements in $\mathcal{A}_{A1,1}$ with keys $(\odd, i, s)$ and $(\even, i, s)$.
An element with each of the keys can be found in $O(\log\min\{\sigma,\log n\})$ time.
Thus, all MUPSs of Type A1-1 can be computed in
$O(\log\min\{\sigma,\log n\}+ \alpha'_{\add})$ time
where $\alpha'_{\add}$ is the number of MUPSs of Type A1-1.
\subsubsection*{Type A1-2.}
The MUPS of Type A1-2 is a contraction of the maximal palindrome in $T'$ centered at $i$.
By definition, there is at most one MUPS of Type A1-2.
\paragraph*{Preprocessing.}

In the preprocessing phase, we again construct $\MA$ and related data structures as in Type A1-1.
Further, we apply the $O(n)$-time preprocessing of Lemma~\ref{lem:lpp_centered_at_i} for Problem~\ref{prob:lpp_centered_at_i}.
The total preprocessing time is $O(n)$.

\paragraph*{Query.}
Given substitution query $\sub(i,s)$, we compute the MUPS centered at $i$ in $T'$ as follows~(if it exists):
It is clear that $T'[i..i] = s$ is the MUPS of Type A1-2 if $s$ is a unique character in $T'$.
In what follows, we consider the other case.
Let $w$ be the maximal palindrome centered at $i$ in $T'$.
First, we compute the maximum lcp value $\ell_w$ between $\Rarm_w$ and extended arms in $\MA_{\odd, i, s}$.
Then, any contraction $y$ of $w$, such that the arm-length of $y$ is at least $\ell_w$, has no occurrences which cover position $i$ in $T'$,
i.e., $|\inbeg_{T', i}(y)| = 1$.
We can compute $\ell_w$ in $O(\log\min\{\sigma,\log n\})$ time, combining LCE queries and binary search.
Note that $\rarm_w$ occurs at $i+1$ in both $T$ and $T'$ while $\Rarm_w$ might be absent from $T$.
Next we compute the arm-length $\ell'_w$ of the shortest palindrome $v$
such that $\centerpal(v) = i$ and $|\xbeg_{T', i}(v)| = 0$, i.e., $v$ is absent from $T$.
Since the contraction $\tilde{v}$ of $w$ of arm-length $\ell'_w-1$ is the longest palindrome
such that $\centerpal(\tilde{v}) = i$ and $\tilde{v}$ occurs in $T$,
we can reduce the problem of computing $\ell'_w$ to Problem~\ref{prob:lpp_centered_at_i}.
Thus, we can compute $\ell'_w$ in $O(\delta(n,\sigma) + (\log\log n)^2)$ time by Lemma~\ref{lem:lpp_centered_at_i}.
Similar to the case of Type A1-1, by Observation~\ref{obs:unique_covering_i},
the contraction $y^\star$ of $x_w$ of arm-length $\max\{\ell_w, \ell'_w\}$ is a MUPS of Type A1-2, if such $y^\star$ exists.
Therefore, the MUPS of Type A1-2 can be computed in $O(\delta(n,\sigma) + (\log\log n)^2)$ time.
\subsubsection*{Type A2.}
A MUPS of Type A2 occurs at least twice in $T$, and there is only one occurrence not covering editing position $i$.
For a palindrome $w$ repeating in $T$,
let $T[b_{l1}.. e_{l1}]$ and $T[b_{l2}.. e_{l2}]$ be the leftmost and the second leftmost occurrence of $w$ in $T$.
Further, let $T[b_{r1}.. e_{r1}]$ and $T[b_{r2}.. e_{r2}]$ be the rightmost and the second rightmost occurrence of $w$ in $T$.
We define interval $\rho_w$ as the intersection of all occurrences of $w$ except for the leftmost one,
i.e., $\rho_w = \{ k \mid k \not\in [b_{l1}, e_{l1}] \text{ and } k \in [b_{r1}, e_{l2}]\}$~(see also Fig.~\ref{fig:type_A2}).
\begin{figure}[tb]
  \centerline{
    \includegraphics[scale=0.35]{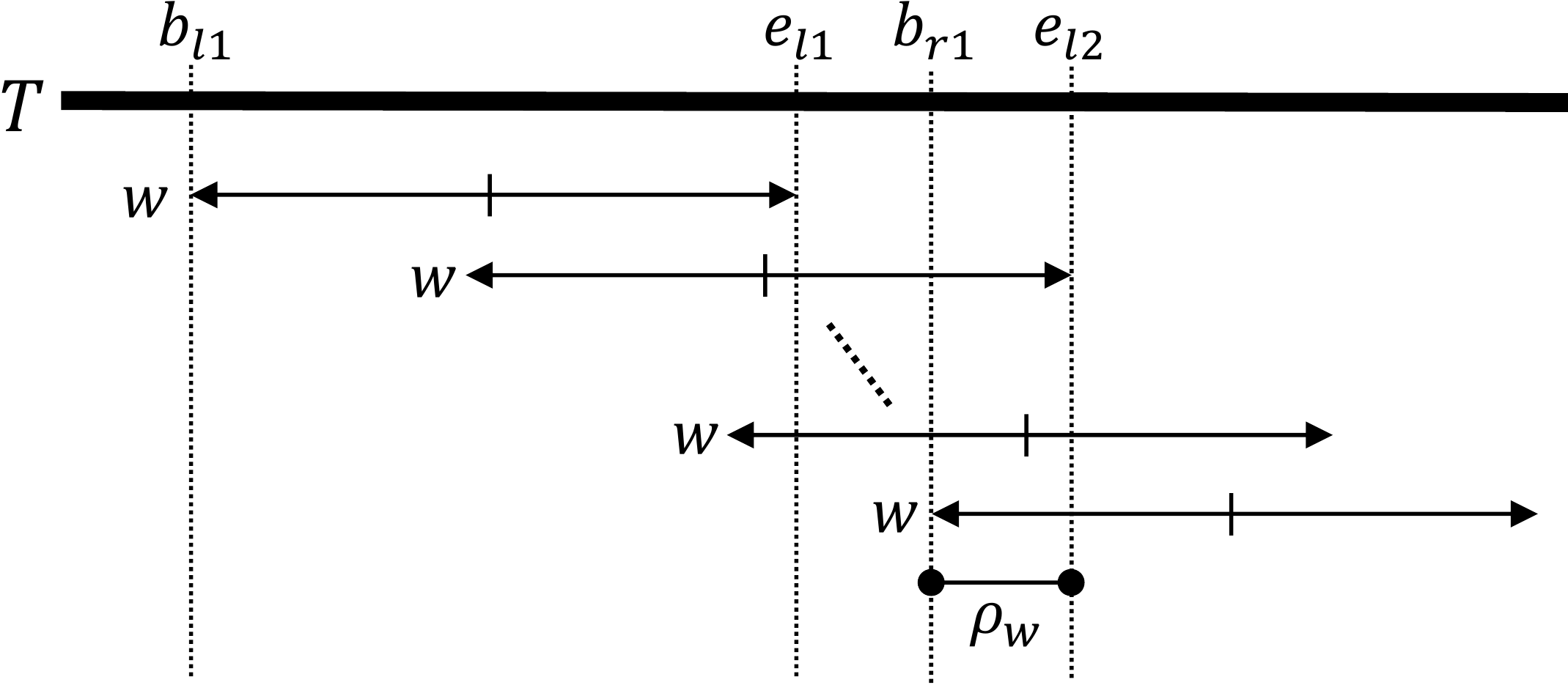}
  }
  \caption{Illustration for $\rho_w$ of Type A2.
  The top two arrows represent the leftmost and the second leftmost occurrence of $w$.
  Also, the bottom two arrows represent the second rightmost and the rightmost occurrence of $w$.
  In this case, $\rho_w = [b_{r1},e_{l2}]$.
  }
    \label{fig:type_A2}
\end{figure}
Similarly, we define interval $\tilde{\rho}_w$ as the intersection of all occurrences of $w$ except for the rightmost one.
Note that $\rho_w$ and $\tilde{\rho}_w$ can be empty.
Then, $w$ is unique after the $i$-th character is edited if and only if $i \in \rho_w \cup \tilde{\rho}_w$.
Thus, any MUPS of Type A2 is a palindrome corresponding to some interval in $\rho_w \cup \tilde{\rho}_w$ stabbed by $i$.
To avoid accessing intervals that do not correspond to the MUPSs to be added, we decompose each $\rho_w$.
It is easy to see that for any contraction $v$ of $w$, $\rho_v \subset \rho_w$ holds.
Also, if $T[i]$, with $i \in \rho_{v}$, is edited,
then both $w$ and $v$ become unique in $T'$, i.e., $w$ cannot be a MUPS of $T'$.
For each unique palindrome $w$ in $T$,
we decompose $\rho_w$ into at most three intervals $\rho_w = \rho^1_w\rho_{w'}\rho^2_w$ where $w' = w[2.. |w|-1]$.
Similarly, we decompose $\tilde{\rho}_w$ into $\tilde{\rho}_w = \tilde{\rho}^1_w\tilde{\rho}_{w'}\tilde{\rho}^2_w$.
Then, $w$ is a MUPS of Type A2 if and only if $i \in \rho^1_w \cup \rho^2_w \cup \tilde{\rho}^1_w \cup \tilde{\rho}^2_w$.

\paragraph*{Preprocessing.}

In the preprocessing phase, we first construct the eertree of $T$ and
the suffix tree of $T$ enhanced with additional explicit nodes for all distinct palindromes in $T$.
Next, we compute at most four (leftmost, second leftmost, rightmost, second rightmost) occurrences of each palindrome in $T$
by traversing the enhanced suffix tree.
At the same time, we compute $\rho_w$ and $\tilde{\rho}_w$ for each palindrome in $w$.
Next, we sequentially access distinct palindromes by traversing $\EERTREE(T)$ in a pre-order manner.
Then, for each palindrome $w$, we decompose $\rho_w$ and $\tilde{\rho_w}$ based on the rules as mentioned above.
Finally, we apply the preprocessing for stabbing queries to the $O(n)$ intervals obtained.
The total preprocessing time is $O(n)$. 

\paragraph*{Query.}
Given a query $\sub(i, s)$, we compute all intervals stabbed by position $i$.
The palindromes corresponding to the intervals are MUPSs of Type A2.
Hence, the query time is $O(1 + \beta_{\add})$, where $\beta_{\add}$ is the number of MUPSs of Type A2.
\subsubsection*{Type A3.}
A MUPS of Type A3 is unique but not minimal in $T$.
Such a unique palindrome $u$ in $T$ contains a MUPS $w \ne u$ of $T$ as a contraction.
Since $u$ is a MUPS of $T'$, $w$ is repeating in $T'$, i.e., $w$ is a removed MUPS of Type R2.
Contrary, consider a MUPS $w$ of Type R2, which is repeating in $T'$.
Then, the shortest unique expansion of $w$ in $T'$ is an added MUPS of Type A3, if it exists.
The preprocessing for Type A3 is almost the same as for Type R2.
We store a bit more information for Type A3 in addition to the information in $\mathcal{A}_{R2}$.
\paragraph*{Preprocessing.}

In the preprocessing phase,
we first apply the $O(n)$-time preprocessing of Lemma~\ref{lem:lpp_centered_at_i} for Problem~\ref{prob:lpp_centered_at_i}.
Next, we initialize the set $\mathcal{A}_{A3} = \emptyset$.
This set $\mathcal{A}_{A3}$ will become an index of MUPSs of Type A3 when the preprocessing is finished.
For each MUPS $w = T[b.. e]$ of $T$, we process the followings:
For the beginning position $j \ne b$ of each occurrence of $\Larm_w$ in $T$,
we compute the lcp value $\ell_j$ between $T[j+|\Larm_w|..|T|]\$$ and $T[\lceil c \rceil..|T|]\$$ with allowing one mismatch
where $c$ is the center of $w$ in $T$.
If $\ell_j$ is smaller than $\rarm_w$, then we do nothing for this occurrence since $w$ cannot occur at $j$ after any single-character substitution.
Otherwise, let $q = j+|\Larm_w|-1+d$ be the first mismatched position in $T$.
When the $q$-th character of $T$ is substituted by the character $\rarm_w[d]$,
$w = T[b.. e]$ occurs at $j \ne b$, i.e., it is a MUPS of Type R2.
Unlike for Type R2,
we add the \emph{pair} of MUPS and (1-mismatched) lcp value $(T[b.. e], \ell_j)$ into $\mathcal{A}_{A3}$
with the pair of index and character $(q, \rarm_w[d])$ as the key.
In addition, symmetrically, we update $\mathcal{A}_{A3}$ for each occurrence of $\Rarm_w$ in $T$.
After finishing the above processes for every MUPS of $T$, we then sort the elements of $\mathcal{A}_{A3}$ by radix sort on the keys.
If there are multiple identical elements with the same key, we unify them into a single element.
Also, if there are multiple elements with the same key, we store them in a linear list.
By Lemma~\ref{lem:different_occ}, the total number of occurrences of arms of MUPSs is $O(n)$,
and hence, the total preprocessing time is $O(n)$.

\paragraph*{Query.}
Given a query $\sub(i, s)$,
we query Problem~\ref{prob:lpp_centered_at_i} with the same pair $(i, s)$ as the input.
Then, we complete checking whether there exists a MUPS of Type A3 centered at $i$.
For the remaining MUPSs of Type 3,
we retrieve the MUPSs of Type A3 using the index $\mathcal{A}_{A3}$
as in the query algorithm for Type R2.
This can be done in $O(\log\min\{\sigma, \log n\} + \gamma_{\add})$ time
where $\gamma_{\add}$ is the number of MUPSs of Type A3.
Therefore, the total query time of Type A3 is $O(\delta(n,\sigma) + (\log\log n)^2 + \gamma_{\add})$.
\\
\\
\noindent
To summarize,
we can compute all MUPSs to be added after a single-character substitution
in $O(\delta(n,\sigma) + (\log\log n)^2 + \alpha'_{\add} + \beta_{\add} + \gamma_{\add})$ time.
Then, combining the results of Section~\ref{subsec:removedMUPS} with the above results, we obtain the following theorem:
\begin{theorem}\label{thm:using_PTLCE}
  After $O(n)$-time preprocessing,
  we can compute the set of MUPSs after a single-character substitution
  in $O(\delta(n, \sigma) + (\log\log n)^2 + d) \subset O(\log n)$ time
  where $d$ is the number of changes of MUPSs.
\end{theorem}

\subsection{Alternative Algorithm for Problem~\ref{prob:lpp_centered_at_i}} \label{subsec:using_NCA}
The query time of Theorem~\ref{thm:using_PTLCE} is dominated by
the time to answer Problem~\ref{prob:lpp_centered_at_i}.
Here, we introduce another solution for Problem~\ref{prob:lpp_centered_at_i}
utilizing \emph{nearest colored ancestor queries} instead of path-tree LCE queries.
\paragraph*{Preprocessing for Problem~\ref{prob:lpp_centered_at_i}.}
We first construct the suffix tree of $T\$$.
Also, for each odd-palindrome in $T$,
we make the locus of the right arm explicit
and label the node with the pair of the center character and the binary flag that indicates if the palindrome is a MUPS.
We regard the pair as the \emph{color} of the node.
Furthermore, we apply a preprocessing for NCA queries to the colored tree\footnote{
  There can be a node with multiple colors in the tree.
  However, we can easily avoid such a situation by copying a node with $k$ colors to $k$ nodes.
  Also, in the case of Problem~\ref{prob:lpp_centered_at_i}, the cumulative total number of colored nodes is $O(n)$.
}.
The preprocessing time is $O(n+c_{\nca}(n,\sigma))$, where $c_{\nca}(n,\sigma)$ is the preprocessing time for NCA queries.

\paragraph*{Query for Problem~\ref{prob:lpp_centered_at_i}.}
Given a substitution query $\sub(i,s)$, we start at the node
corresponding to $\rarm_w$ where $w$ is the maximal palindrome in $T$ centered at $i$.
We then compute the nearest ancestor $V$ colored with $(s, \mathsf{0})$ by using NCA query.
If such node $V$ exists, palindrome $P = \str(V)^R\cdot s\cdot\str(V)$ is the answer of the former part of Problem~\ref{prob:lpp_centered_at_i}
where $\str(V)$ denotes the string corresponding to $V$ in the enhanced suffix tree of $T\$$.
Also, we query NCA $(s, \mathsf{1})$ from $V$.
We can determine if $P$ is unique, and if it is unique, we can find the MUPS contained in $P$.
The query time is $O(q_{\nca}(n,\sigma))$
where $q_{\nca}(n,\sigma)$ is the query time for NCA.

Let $s_{\nca}(n,\sigma)$ denote the space for the NCA data structure.
We obtain the following theorem:
\begin{theorem}\label{thm:using_NCA}
  After $O(n + c_{\nca}(n,\sigma))$-time and $O(n + s_{\nca}(n,\sigma))$-space preprocessing,
  we can compute the set of MUPSs after a single-character substitution in $O(q_{\nca}(n,\sigma) + \log\min\{\sigma,\log n\} + d)$ time.
\end{theorem}

The results for NCA queries in Lemmas~\ref{lem:nca_general} and \ref{lem:nca_log_colors}
can be plugged into the functions $c_\nca$, $q_\nca$, and $s_\nca$.
In addition, even when a general case,
we can handle $\delta(n,\sigma)$ as a constant by utilizing a perfect hashing~\cite{Fredman1984hashing}
after $O(n\loglog n)$-time or $O(n)$-expected time preprocessing.
Table~\ref{tab:nca} lists different representations of the time/space complexities
of Theorems~\ref{thm:using_PTLCE} and \ref{thm:using_NCA}.
We emphasize that our algorithm runs in optimal $O(1+d)$ time when $\sigma$ is constant.
\begin{corollary}
  If $\sigma \in O(1)$,
  after $O(n)$-time and $O(n)$-preprocessing,
  we can compute the set of MUPSs after a single-character substitution in $O(1+d)$ time.
\end{corollary}

\begin{table}[t]
\centering
\begin{tabular}{l@{\hspace{1em}}l@{\hspace{1em}}l@{\hspace{1em}}l@{\hspace{1em}}l}
\toprule
            & \multicolumn{2}{c}{Time}\\
\cmidrule{2-3}
$\sigma$    & Construction    & Query                              & Space    & Ref.\\
\midrule
$n^{O(1)}$  & $O(n)$          & $O(\log\sigma + (\loglog n)^2+ d)$ & $O(n)$   & Theorem~\ref{thm:using_PTLCE}, \ref{thm:pathtreeLCE}\\
$n^{O(1)}$  & $O(n\loglog n)$ & $O(\loglog n + d)$                 & $O(n)$   & Theorem~\ref{thm:using_NCA}, Lemma~\ref{lem:nca_general}\\
$n^{O(1)}$  & expected $O(n)$ & $O(\loglog n + d)$                 & $O(n)$   & Theorem~\ref{thm:using_NCA}, Lemma~\ref{lem:nca_general}\\
$O(n)$      & $O(n)$          & $O((\loglog n)^2 + d)$             & $O(n)$   & Theorem~\ref{thm:using_PTLCE}, \ref{thm:pathtreeLCE}\\
$O(\log n)$ & $O(n)$          & $O(\loglog n + d)$                 & $O(n)$   & Theorem~\ref{thm:using_NCA}, Lemma~\ref{lem:nca_log_colors}\\
$O(1)$      & $O(n)$          & $O(1+d)$                             & $O(n)$   & Theorem~\ref{thm:using_NCA}, Lemma~\ref{lem:nca_log_colors}\\
\bottomrule
\end{tabular}
\caption{
  Concrete complexities of our algorithms for the problem of computing MUPSs after a single-character substitution.
  All the above results require only linear space.
  Each query time is $O(\log n)$ since $\log\sigma \in O(\log n)$ and $d \in O(\log n)$.
} \label{tab:nca}
\end{table}
 \section{Conclusions and Future Work} \label{sec:conclusions}

In this paper, we dealt with the problem of updating the set of MUPSs of a string after a single-character substitution.
We showed that the number $d$ of changes of MUPSs after a single-character substitution is $O(\log n)$.
Furthermore, we presented an algorithm that uses $O(n)$ time and space for preprocessing,
and updates the set of MUPSs in $O(\log\sigma + (\loglog n)^2 + d)$ time where $\sigma$ is the alphabet size.
We also proposed a variant of the algorithm, which runs in optimal $O(1+d)$ time when the alphabet size is constant.

Our future work includes the following:
\begin{itemize}
  \item[(1)] Can our algorithm be adapted to the cases of insertions and deletions?
  \item[(2)] Can we extend our algorithm to a fully dynamic setting?
  It is interesting whether the techniques in~\cite{AmirCPR19,DBLP:journals/corr/abs-1906-09732} can be utilized in the dynamic version of the MUPS problem. 
\end{itemize}
 \section*{Acknowledgements}
We would like to thank Associate Professor Shunsuke Inenaga (Kyushu University)
for the valuable discussions on simplifying our algorithms.
This work was supported by the JSPS KAKENHI Grant Numbers JP20J21147 (MF) and JP20J11983 (TM).


\clearpage
\begin{thebibliography}{10}

\bibitem{AbedinH0T18}
Paniz Abedin, Sahar Hooshmand, Arnab Ganguly, and Sharma~V. Thankachan.
\newblock The heaviest induced ancestors problem revisited.
\newblock In Gonzalo Navarro, David Sankoff, and Binhai Zhu, editors, {\em
  Annual Symposium on Combinatorial Pattern Matching, {CPM} 2018, July 2-4,
  2018 - Qingdao, China}, volume 105 of {\em LIPIcs}, pages 20:1--20:13.
  Schloss Dagstuhl - Leibniz-Zentrum f{\"{u}}r Informatik, 2018.
\newblock \href {https://doi.org/10.4230/LIPIcs.CPM.2018.20}
  {\path{doi:10.4230/LIPIcs.CPM.2018.20}}.

\bibitem{DBLP:journals/corr/abs-1906-09732}
Amihood Amir and Itai Boneh.
\newblock Dynamic palindrome detection.
\newblock {\em CoRR}, abs/1906.09732, 2019.
\newblock URL: \url{http://arxiv.org/abs/1906.09732}, \href
  {http://arxiv.org/abs/1906.09732} {\path{arXiv:1906.09732}}.

\bibitem{AmirBCK19}
Amihood Amir, Itai Boneh, Panagiotis Charalampopoulos, and Eitan Kondratovsky.
\newblock Repetition detection in a dynamic string.
\newblock In Michael~A. Bender, Ola Svensson, and Grzegorz Herman, editors,
  {\em 27th Annual European Symposium on Algorithms, {ESA} 2019, September
  9-11, 2019, Munich/Garching, Germany}, volume 144 of {\em LIPIcs}, pages
  5:1--5:18. Schloss Dagstuhl - Leibniz-Zentrum f{\"{u}}r Informatik, 2019.
\newblock \href {https://doi.org/10.4230/LIPIcs.ESA.2019.5}
  {\path{doi:10.4230/LIPIcs.ESA.2019.5}}.

\bibitem{AmirCIPR17}
Amihood Amir, Panagiotis Charalampopoulos, Costas~S. Iliopoulos, Solon~P.
  Pissis, and Jakub Radoszewski.
\newblock Longest common factor after one edit operation.
\newblock In Gabriele Fici, Marinella Sciortino, and Rossano Venturini,
  editors, {\em String Processing and Information Retrieval - 24th
  International Symposium, {SPIRE} 2017, Palermo, Italy, September 26-29, 2017,
  Proceedings}, volume 10508 of {\em Lecture Notes in Computer Science}, pages
  14--26. Springer, 2017.
\newblock \href {https://doi.org/10.1007/978-3-319-67428-5\_2}
  {\path{doi:10.1007/978-3-319-67428-5\_2}}.

\bibitem{AmirCPR19}
Amihood Amir, Panagiotis Charalampopoulos, Solon~P. Pissis, and Jakub
  Radoszewski.
\newblock Longest common substring made fully dynamic.
\newblock In Michael~A. Bender, Ola Svensson, and Grzegorz Herman, editors,
  {\em 27th Annual European Symposium on Algorithms, {ESA} 2019, September
  9-11, 2019, Munich/Garching, Germany}, volume 144 of {\em LIPIcs}, pages
  6:1--6:17. Schloss Dagstuhl - Leibniz-Zentrum f{\"{u}}r Informatik, 2019.
\newblock \href {https://doi.org/10.4230/LIPIcs.ESA.2019.6}
  {\path{doi:10.4230/LIPIcs.ESA.2019.6}}.

\bibitem{Apostolico1995parallel}
Alberto Apostolico, Dany Breslauer, and Zvi Galil.
\newblock Parallel detection of all palindromes in a string.
\newblock {\em Theor. Comput. Sci.}, 141(1{\&}2):163--173, 1995.
\newblock \href {https://doi.org/10.1016/0304-3975(94)00083-U}
  {\path{doi:10.1016/0304-3975(94)00083-U}}.

\bibitem{Bille2016LCE}
Philip Bille, Pawel Gawrychowski, Inge~Li G{\o}rtz, Gad~M. Landau, and Oren
  Weimann.
\newblock Longest common extensions in trees.
\newblock {\em Theor. Comput. Sci.}, 638:98--107, 2016.
\newblock \href {https://doi.org/10.1016/j.tcs.2015.08.009}
  {\path{doi:10.1016/j.tcs.2015.08.009}}.

\bibitem{Bille2015GrammerCompressedString}
Philip Bille, Gad~M. Landau, Rajeev Raman, Kunihiko Sadakane, Srinivasa~Rao
  Satti, and Oren Weimann.
\newblock Random access to grammar-compressed strings and trees.
\newblock {\em {SIAM} J. Comput.}, 44(3):513--539, 2015.
\newblock \href {https://doi.org/10.1137/130936889}
  {\path{doi:10.1137/130936889}}.

\bibitem{Charalampopoulos20}
Panagiotis Charalampopoulos, Pawel Gawrychowski, and Karol Pokorski.
\newblock Dynamic longest common substring in polylogarithmic time.
\newblock In Artur Czumaj, Anuj Dawar, and Emanuela Merelli, editors, {\em 47th
  International Colloquium on Automata, Languages, and Programming, {ICALP}
  2020, July 8-11, 2020, Saarbr{\"{u}}cken, Germany (Virtual Conference)},
  volume 168 of {\em LIPIcs}, pages 27:1--27:19. Schloss Dagstuhl -
  Leibniz-Zentrum f{\"{u}}r Informatik, 2020.
\newblock \href {https://doi.org/10.4230/LIPIcs.ICALP.2020.27}
  {\path{doi:10.4230/LIPIcs.ICALP.2020.27}}.

\bibitem{Charalampopoulos2021Internal}
Panagiotis Charalampopoulos, Tomasz Kociumaka, Manal Mohamed, Jakub
  Radoszewski, Wojciech Rytter, and Tomasz Wale{\'n}.
\newblock Internal dictionary matching.
\newblock {\em Algorithmica}, 2021.
\newblock \href {https://doi.org/10.1007/s00453-021-00821-y}
  {\path{doi:10.1007/s00453-021-00821-y}}.

\bibitem{DroubayJP01}
Xavier Droubay, Jacques Justin, and Giuseppe Pirillo.
\newblock Episturmian words and some constructions of de {Luca} and {Rauzy}.
\newblock {\em Theor. Comput. Sci.}, 255(1-2):539--553, 2001.
\newblock \href {https://doi.org/10.1016/S0304-3975(99)00320-5}
  {\path{doi:10.1016/S0304-3975(99)00320-5}}.

\bibitem{Farach-Colton2000suffixtree}
Martin Farach{-}Colton, Paolo Ferragina, and S.~Muthukrishnan.
\newblock On the sorting-complexity of suffix tree construction.
\newblock {\em J. {ACM}}, 47(6):987--1011, 2000.
\newblock \href {https://doi.org/10.1145/355541.355547}
  {\path{doi:10.1145/355541.355547}}.

\bibitem{Fredman1984hashing}
Michael~L. Fredman, J{\'{a}}nos Koml{\'{o}}s, and Endre Szemer{\'{e}}di.
\newblock Storing a sparse table with {$O(1)$} worst case access time.
\newblock {\em J. {ACM}}, 31(3):538--544, 1984.
\newblock \href {https://doi.org/10.1145/828.1884}
  {\path{doi:10.1145/828.1884}}.

\bibitem{FunakoshiM21MUPSEdit}
Mitsuru Funakoshi and Takuya Mieno.
\newblock Minimal unique palindromic substrings after single-character
  substitution.
\newblock In Thierry Lecroq and H{\'{e}}l{\`{e}}ne Touzet, editors, {\em String
  Processing and Information Retrieval - 28th International Symposium, {SPIRE}
  2021, Lille, France, October 4-6, 2021, Proceedings}, volume 12944 of {\em
  Lecture Notes in Computer Science}, pages 33--46. Springer, 2021.
\newblock \href {https://doi.org/10.1007/978-3-030-86692-1\_4}
  {\path{doi:10.1007/978-3-030-86692-1\_4}}.

\bibitem{Funakoshi2021longestpal_afteredit}
Mitsuru Funakoshi, Yuto Nakashima, Shunsuke Inenaga, Hideo Bannai, and Masayuki
  Takeda.
\newblock Computing longest palindromic substring after single-character or
  block-wise edits.
\newblock {\em Theor. Comput. Sci.}, 859:116--133, 2021.
\newblock \href {https://doi.org/10.1016/j.tcs.2021.01.014}
  {\path{doi:10.1016/j.tcs.2021.01.014}}.

\bibitem{GasieniecSWAT96}
Leszek Gasieniec, Marek Karpinski, Wojciech Plandowski, and Wojciech Rytter.
\newblock Efficient algorithms for {Lempel}-{Ziv} encoding.
\newblock In Rolf Karlsson and Andrzej Lingas, editors, {\em Algorithm Theory
  --- SWAT'96}, pages 392--403, Berlin, Heidelberg, 1996. Springer Berlin
  Heidelberg.

\bibitem{GawrychowskiIIK18}
Pawel Gawrychowski, Tomohiro I, Shunsuke Inenaga, Dominik K{\"{o}}ppl, and
  Florin Manea.
\newblock Tighter bounds and optimal algorithms for all maximal
  {\(\alpha\)}-gapped repeats and palindromes - finding all maximal
  {\(\alpha\)}-gapped repeats and palindromes in optimal worst case time on
  integer alphabets.
\newblock {\em Theory Comput. Syst.}, 62(1):162--191, 2018.
\newblock \href {https://doi.org/10.1007/s00224-017-9794-5}
  {\path{doi:10.1007/s00224-017-9794-5}}.

\bibitem{Gawrychowski2018NCA}
Pawel Gawrychowski, Gad~M. Landau, Shay Mozes, and Oren Weimann.
\newblock The nearest colored node in a tree.
\newblock {\em Theor. Comput. Sci.}, 710:66--73, 2018.
\newblock \href {https://doi.org/10.1016/j.tcs.2017.08.021}
  {\path{doi:10.1016/j.tcs.2017.08.021}}.

\bibitem{inoue2018algorithms}
Hiroe Inoue, Yuto Nakashima, Takuya Mieno, Shunsuke Inenaga, Hideo Bannai, and
  Masayuki Takeda.
\newblock Algorithms and combinatorial properties on shortest unique
  palindromic substrings.
\newblock {\em J. Discrete Algorithms}, 52-53:122--132, 2018.
\newblock \href {https://doi.org/10.1016/j.jda.2018.11.009}
  {\path{doi:10.1016/j.jda.2018.11.009}}.

\bibitem{Kociumaka2020seedscomputation}
Tomasz Kociumaka, Marcin Kubica, Jakub Radoszewski, Wojciech Rytter, and Tomasz
  Walen.
\newblock A linear-time algorithm for seeds computation.
\newblock {\em {ACM} Trans. Algorithms}, 16(2):27:1--27:23, 2020.
\newblock \href {https://doi.org/10.1145/3386369} {\path{doi:10.1145/3386369}}.

\bibitem{Manacher75}
Glenn~K. Manacher.
\newblock A new linear-time "on-line" algorithm for finding the smallest
  initial palindrome of a string.
\newblock {\em J. {ACM}}, 22(3):346--351, 1975.
\newblock \href {https://doi.org/10.1145/321892.321896}
  {\path{doi:10.1145/321892.321896}}.

\bibitem{matsubara_tcs2009}
Wataru Matsubara, Shunsuke Inenaga, Akira Ishino, Ayumi Shinohara, Tomoyuki
  Nakamura, and Kazuo Hashimoto.
\newblock Efficient algorithms to compute compressed longest common substrings
  and compressed palindromes.
\newblock {\em Theor. Comput. Sci.}, 410(8-10):900--913, 2009.
\newblock \href {https://doi.org/10.1016/j.tcs.2008.12.016}
  {\path{doi:10.1016/j.tcs.2008.12.016}}.

\bibitem{Mieno2020EERTREE}
Takuya Mieno, Kiichi Watanabe, Yuto Nakashima, Shunsuke Inenaga, Hideo Bannai,
  and Masayuki Takeda.
\newblock Palindromic trees for a sliding window and its applications.
\newblock {\em Information Processing Letters}, 173:106174, 2022.
\newblock \href {https://doi.org/https://doi.org/10.1016/j.ipl.2021.106174}
  {\path{doi:https://doi.org/10.1016/j.ipl.2021.106174}}.

\bibitem{Rubinchik2018eertree}
Mikhail Rubinchik and Arseny~M. Shur.
\newblock {EERTREE:} an efficient data structure for processing palindromes in
  strings.
\newblock {\em Eur. J. Comb.}, 68:249--265, 2018.
\newblock \href {https://doi.org/10.1016/j.ejc.2017.07.021}
  {\path{doi:10.1016/j.ejc.2017.07.021}}.

\bibitem{RubinchikS20}
Mikhail Rubinchik and Arseny~M. Shur.
\newblock Palindromic k-factorization in pure linear time.
\newblock In Javier Esparza and Daniel Kr{\'{a}}l', editors, {\em 45th
  International Symposium on Mathematical Foundations of Computer Science,
  {MFCS} 2020, August 24-28, 2020, Prague, Czech Republic}, volume 170 of {\em
  LIPIcs}, pages 81:1--81:14. Schloss Dagstuhl - Leibniz-Zentrum f{\"{u}}r
  Informatik, 2020.
\newblock \href {https://doi.org/10.4230/LIPIcs.MFCS.2020.81}
  {\path{doi:10.4230/LIPIcs.MFCS.2020.81}}.

\bibitem{Schmidt2009stabbing}
Jens~M. Schmidt.
\newblock Interval stabbing problems in small integer ranges.
\newblock In Yingfei Dong, Ding{-}Zhu Du, and Oscar~H. Ibarra, editors, {\em
  Algorithms and Computation, 20th International Symposium, {ISAAC} 2009,
  Honolulu, Hawaii, USA, December 16-18, 2009. Proceedings}, volume 5878 of
  {\em Lecture Notes in Computer Science}, pages 163--172. Springer, 2009.
\newblock \href {https://doi.org/10.1007/978-3-642-10631-6\_18}
  {\path{doi:10.1007/978-3-642-10631-6\_18}}.

\bibitem{UrabeNIBT18}
Yuki Urabe, Yuto Nakashima, Shunsuke Inenaga, Hideo Bannai, and Masayuki
  Takeda.
\newblock Longest {Lyndon} substring after edit.
\newblock In Gonzalo Navarro, David Sankoff, and Binhai Zhu, editors, {\em
  Annual Symposium on Combinatorial Pattern Matching, {CPM} 2018, July 2-4,
  2018 - Qingdao, China}, volume 105 of {\em LIPIcs}, pages 19:1--19:10.
  Schloss Dagstuhl - Leibniz-Zentrum f{\"{u}}r Informatik, 2018.
\newblock \href {https://doi.org/10.4230/LIPIcs.CPM.2018.19}
  {\path{doi:10.4230/LIPIcs.CPM.2018.19}}.

\bibitem{WatanabeNIBT20}
Kiichi Watanabe, Yuto Nakashima, Shunsuke Inenaga, Hideo Bannai, and Masayuki
  Takeda.
\newblock Fast algorithms for the shortest unique palindromic substring problem
  on run-length encoded strings.
\newblock {\em Theory Comput. Syst.}, 64(7):1273--1291, 2020.
\newblock \href {https://doi.org/10.1007/s00224-020-09980-x}
  {\path{doi:10.1007/s00224-020-09980-x}}.

\bibitem{Weiner1973suffixtree}
Peter Weiner.
\newblock Linear pattern matching algorithms.
\newblock In {\em 14th Annual Symposium on Switching and Automata Theory, Iowa
  City, Iowa, USA, October 15-17, 1973}, pages 1--11. {IEEE} Computer Society,
  1973.
\newblock \href {https://doi.org/10.1109/SWAT.1973.13}
  {\path{doi:10.1109/SWAT.1973.13}}.

\end{thebibliography}
\end{document}